\documentclass{article}
\usepackage[utf8]{inputenc}
\usepackage{natbib}
\usepackage{geometry}  	
\usepackage{amsmath}

\DeclareMathOperator*{\argmax}{arg\,max}

\usepackage{graphicx}
\graphicspath{ {images/} }
\usepackage{amssymb}
\usepackage{amsthm}
\usepackage{color}
\usepackage{enumitem}
\usepackage{setspace}

\title{A Model of Justification}
\author{Sarah Ridout\thanks{ridout@g.harvard.edu. For helpful comments and suggestions, I am indebted to Christine Exley, Jerry Green, Shengwu Li, and especially Matthew Rabin and Tomasz Strzalecki.}}
\date{\today}

\newtheorem{axiom}{Axiom}
\newtheorem{definition}{Definition}[section]
\newtheorem{lemma}{Lemma}[section]
\newtheorem{proposition}{Proposition}[section]
\newtheorem{corollary}{Corollary}[section]
\newtheorem{theorem}{Theorem}[section]

\newtheorem{example}{Example}[section]

\begin{document}

\onehalfspacing

\maketitle

\begin{abstract}
I consider decision-making constrained by considerations of morality, rationality, or other virtues. The decision maker (DM) has a true preference over outcomes, but feels compelled to choose among outcomes that are top-ranked by some preference that he considers “justifiable.” This model unites a broad class of empirical work on distributional preferences, charitable donations, prejudice/discrimination, and corruption/bribery. I provide a behavioral characterization of the model. I also show that the set of justifications can be identified from choice behavior when the true preference is known, and that choice behavior substantially restricts both the true preference and justifications when neither is known. I argue that the justifiability model represents an advancement over existing models of rationalization because the structure it places on possible “rationales” improves tractability, interpretation and identification. 
\end{abstract}

\section{Introduction}

When a decision-maker believes that his choices reflect on his character---his virtues and values---his behavior may be inconsistent with maximization of a stable preference relation. This is not because he lacks a stable preference relation, but because he fears that aspects of his preferences are unjustifiable. To mask these unjustifiable preferences, he restricts himself to choices that would be made by good and reasonable people. Apparent inconsistencies in his preferences arise when these ``good and reasonable people'' disagree, and the decision maker breaks ties in accordance with his true preference. Thus, he appeals to different rationales in different situations.

\citet{exley2016excusing} provides a motivating example. The domain of choice consists of prizes paid to the subject and donations to a charity. Each subject is asked to make three types of choices: between a random prize and a sure prize, between a random donation and a sure donation, and between a random prize and a random donation. \citeauthor{exley2016excusing} finds that subjects penalize risk when the donation is risky relative to the prize, and favor risk when the payment is risky relative to the donation. This behavior is inconsistent with any stable risk preference, but it is well explained by the model of justification studied in this paper. Intuitively, subjects believe that it is wrong to refuse when offered the opportunity to donate on favorable terms, but good and reasonable people may disagree on which terms are ``favorable'' when risk is involved. This ambiguity allows subjects who are not particularly interested in donating to charity to reject donation opportunities without feeling, or looking, too selfish. Section \ref{sec:empirics} reviews several other examples in different domains, including sharing with others, discrimination in hiring, and bribery. 

The main results in this paper fall into two categories. The first category takes the true preference as a primitive of the model, and shows how to identify the preferences that the decision-maker (henceforth DM) considers justifiable. There are two reasons for assuming that the true preference is observable. First, subjects may be willing to reveal their true preference when their choices are anonymous, or implemented by someone else. Some experiments lend support to this idea. For instance, \citet{hamman2010self} found that subjects made blatantly selfish decisions when those decisions were carried out by an agent. Second, it is easier to understand the results for the unknown-preference case after covering the known-preference case, because the results for the more complicated case build on the results for the simpler one. 

Theorem \ref{thm:main} is the first representation theorem in the paper. Given the true preference, it provides necessary and sufficient conditions for behavior to be consistent with the justifiability model. Here, a ``justification'' is a complete and transitive relation on the same domain as the true preference. The key axiom is Irrelevance of Unjustifiable Alternatives. An alternative $a$ is unjustifiable in the presence of a set $A$ if $a$ is weakly preferred to everything selected from $\{a\} \cup A$, but $a$ is not itself selected. The axiom says that $a$ is irrelevant whenever $A$ is present: adding or removing $a$ has no effect on choice. Intuitively, this is because every justifiable preference strictly prefers at least one element of $A$ to $a$. The proof of the theorem shows that the (maximal) set of justifiable preferences is the set of complete and transitive relations consistent with all these restrictions. Proposition \ref{prop:dominance} extends Theorem \ref{thm:main} to account for an exogenous dominance relation on the domain. It provides necessary and sufficient conditions for a representation in which all the justifications respect dominance. The relevant notion of ``dominance'' is very inclusive: it can capture impartiality or anti-discrimination requirements as well as stochastic dominance. 

Theorem \ref{thm:main} does not require the domain to have any particular structure, or restrict its cardinality. It is completely general. Although this is convenient from a theoretical point of view, applications may demand more structure. To this end, Theorem \ref{thm:EU} axiomatizes the expected-utility special case of the justifiability model. Given that the true preference has an expected-utility representation, Theorem \ref{thm:EU} provides necessary and sufficient conditions for the justifiable preferences to have expected-utility representations as well. The proof of Theorem \ref{thm:EU} explains how to identify the set of justifiable utilities. This version of the model is quite easy to work with because of the linear structure. As shown in Proposition \ref{cor:comp_stat_EU}, it also admits nice comparative statics. To determine whether one DM has stricter morals than another, it is only necessary to compare their choices on binary menus with one element fixed. This may be useful to experimenters who run different treatments designed to strengthen or weaken moral considerations in subjects' minds. If the justifiability model is correct, then the sets of justifiable preferences backed out from a fixed subject's behavior across treatments should be nested. 

Section \ref{sec:unobserved} dispenses with the assumption that the true preference is observed. It begins with the EU case, in which identification of the true preference is particularly simple. Typically, the true preference will be uniquely pinned down.\footnote{The ``untypical'' case is a justifiability model with only one justifiable preference. Although the true preference is not uniquely pinned down, it is still possible to pin down a set of candidates. If the axioms fail conditional on any of these candidates, the model is falsified.} Proposition \ref{prop:unobserved_EU} provides a procedure for generating a candidate true preference. The EU justifiability model fits the data if and only if this procedure yields a unique candidate and, conditional on this candidate, the data satisfies the axioms in Theorem \ref{thm:EU}. Thus, Theorem \ref{thm:EU} is useful even if the true preference is not observed.

The rest of Section \ref{sec:unobserved} covers the general model. It provides two different behavioral characterizations, Corollary \ref{cor:unknown_true_pref} and Theorem \ref{thm:unknown_true_pref}. As these characterizations demonstrate, the model imposes substantial restrictions even if the true preference is a ``free parameter;'' the model can be falsified with as few as three elements. The two characterizations complement one another. Corollary \ref{cor:unknown_true_pref} clarifies the restrictions behavior imposes on the true preference, and Theorem \ref{thm:unknown_true_pref} clarifies the restrictions on the justifications. Remarkably, two simple and easily spotted patterns of behavior deliver a full understanding of the justifications available to the DM. Section \ref{sec:double} extends the model to account for multiple decision environments, ranked by the pressure on the DM to find a good justification. Proposition \ref{prop:double} extends Theorem \ref{thm:unknown_true_pref} this case. It provides conditions for the DM's behavior to be consistent with a stable true preference, and for the set of justifications to shrink as the pressure rises. 

Section \ref{sec:theory} relates the justifiability model to existing theoretical work. Two strands of literature deserve particular emphasis. The first strand studies two-tiered models in which one preference is used to break ties in another. \citet{gul2005revealed} and \citet{manzini2007sequentially} are notable examples. Although these models have different motivations and interpretations than the justifiability model, they formally correspond to the special case in which the set of justifiable preferences is a singleton. The second strand discusses rationalization. \citet{cherepanov2013rationalization} is the most closely related paper in this class. That paper and this one have similar motivations, but the ``rationales'' there have much less structure than the ``justifications'' here. They are not required to be complete or transitive, to respect any dominance relation, or to satisfy independence. Section \ref{sec:theory} shows that this additional structure leads to stronger predictions and identification as well as more tractability and easier interpretation. 

\section{Empirical Evidence}
\label{sec:empirics}

The key goal of this paper is to unify and provide a deeper understanding of a variety of empirical papers. This section applies the justifiability model to the key results from these papers.

We begin with a classic paper in the moral-wiggle-room literature, \citet{snyder1979avoidance}. Subjects were asked to choose between two rooms with movie projectors. A person in a wheelchair was present in one room, and the other room was empty. Subjects typically chose the empty room when different movies were playing in the two rooms, but not when the same movie was playing in both. To formalize this, let
\begin{align*}
a_1 &= \text{movie 1 alone} \\
b_1 &= \text{movie 1 with disabled stranger} \\
a_2 &= \text{movie 2 alone}
\end{align*}
Consider a subject who chooses as follows:
\[c(\{a_1, a_2\}) = a_1 \quad c(\{a_1, b_1\}) = b_1 \quad c(\{b_1, a_2\}) = a_2.\]
These choices are cyclic, so are inconsistent with maximization of any stable preference relation. The justifiability model explains these as follows. The DM does have a stable preference, given by $a_1 \succ a_2 \succ b_1$. He maximizes this preference over the subset of alternatives that (in his opinion) might be selected by a good and reasonable person. Formally, let $\mathcal{M}$ be the set of preferences that (in his opinion) a good and reasonable person might have. Since the DM has $a_1 \succ b_1$ but $b_1 = c(\{a_1, b_1\})$, he must believe that every $\succsim_m \in \mathcal{M}$ has $b_1 \succ_m a_1$. Taking $\mathcal{M}$ to be the set of preferences that satisfy this constraint, we correctly predict the DM's behavior. When the disabled stranger is not present, the DM can choose his favorite movie because the preference given by $b_1 \succ_m a_1 \succ_m a_2$ is justifiable. When his favorite movie is playing in the room with the stranger, but not the other room, the DM can avoid the stranger because the preference given by $a_2 \succ_m b_1 \succ_m a_1$ is justifiable. But when the same movie is playing in both rooms, the DM feels compelled to socialize with the stranger because all justifiable preferences have $b_1 \succ_m a_1$. We return to this example in Section \ref{sec:unobserved}.

\citet{norton2004casuistry} is a more recent paper in a similar tradition. Subjects were asked to choose between a male and a female candidate for a traditionally male role. In one treatment, the male had more education and the female had more experience; in the other, these attributes were flipped. Subjects exhibited a preference for the male candidate in both cases. However, they denied that gender played a role in their decisions, instead citing the attribute (education or experience) in which the male candidate was superior. These results fit into the same framework as above. Slightly modifying the experiment to break ties between equally qualified candidates, we have
\begin{align*}
a_1 &= \text{male, high education, low experience} \\
b_1 &= \text{female, high education + }\epsilon, \text{ low experience} \\
a_2 &= \text{male, low education, high experience}
\end{align*}
A subject who believes that education is more important than experience, but above all prefers a male candidate, will choose as above. The same model can be used to explain observed choice. Only the interpretation is different. Here, $\mathcal{M}$ is the set of preferences that prefer a more-qualified candidate over a less-qualified candidate, regardless of gender.

Now we proceed to more familiar economic domains. In all of the following experiments, subjects' choices affected others' monetary payoffs as well as their own. All of the results can be explained with a true preference that cares only (or mostly) about own payoff, and a set of ``impartial'' justifiable preferences. The justifiable preferences are discussed in more detail below. 

\citet{gneezy2016motivated} conducted two experiments in which subjects were offered bribes for particular choices. In the first experiment, subjects were told to select the funniest joke from a set of jokes written by other participants. Each joke writer could attempt to bribe the chooser by offering to pay if his joke was selected. In the second experiment, subjects were placed in the role of investment advisors. They were asked to make an investment recommendation to another participant, who was not informed about the assets in the choice set. They were also offered a bribe to recommend a particular asset. \citet{gneezy2016motivated} found that the bribes mattered in both cases, but only if the subjects were shown the bribes \textit{before} evaluating their options. Apparently, subjects felt able to choose the option that came with the larger bribe only if they were able to convince themselves that it was genuinely superior. To explain this within the justifiability model, take $\mathcal{M}$ to be a set of preferences that are independent of the bribe, but exhibit different senses of humor (for experiment 1) or risk attitudes (for experiment 2). 

\citet{rodriguez2012self} presented subjects with an allocation problem in which the money to be allocated was ``earned'' in an earlier phase of the experiment. Subjects earned money by completing a multiple-choice test. After the test was complete, each subject's correct answers were converted into money at a random ``wage.'' The wage was fixed for each subject, but could differ across subjects. Then, subjects were paired up, and one member of each pair was asked to allocate the money earned by the pair.\citeauthor{rodriguez2012self} found that subjects who had low wages and accuracy tended to favor an equal division of money, while subjects who had high wages and accuracy favored division on the basis of earnings. Finally, subjects who had low wages but high accuracy favored division on the basis of accuracy alone (correcting for the unequal ``wages'' assigned by the experimenter). To explain this within the justifiability model, let $\mathcal{M}$ be a set of preferences that treat self and other the same way, but disagree on the way accuracy and earnings should be rewarded. 

As mentioned in the introduction, \citet{exley2016excusing} investigated the effect of risk on allocation choices. Subjects kept more for themselves when the charitable donation was risky and their own payoff was certain \textit{and} when the donation was certain and their own payoff was risky. This is inconsistent with maximization of any preference, let alone an expected-utility preference. It is easily explained within the expected-utility version of the justifiability model, in which both the true preference and all the justifiable preferences take an expected-utility form. (This model is particularly neat and tractable, so Section \ref{sec:EU} is devoted to it.) All the justifiable preferences are generous---they weakly prefer $\$1$ to charity over $\$1$ to self---but they embody different risk attitudes.

\citet{haisley2010self} conducted a similar experiment, but considered uncertainty as well as risk. Subjects chose a relatively equitable allocation that gave moderate prizes to the subject and a stranger over a relatively unequitable one that gave a large prize to the subject and a lottery between small prizes to the stranger. This pattern was reversed when the risk in the lottery over small prizes was replaced with uncertainty (even though the prizes themselves remained unattractive). \citeauthor{haisley2010self} rejected ambiguity loving as an explanation for this reversal. When subjects were asked to choose between the risky lottery and the ambiguous prospect, they chose the former. Moreover, subjects who were confronted with this choice immediately before the choice in the main experiment tended \textit{not} to exhibit the pattern above. To explain these results within the justifiability model, take $\mathcal{M}$ to be a set of preferences that treat self and other the same way, but disagree on ambiguity attitude. A subject confronted with the choice in the main experiment can keep more for himself by pretending to be ambiguity-loving. This option goes away when he has already revealed that he is ambiguity-averse.

In both \citet{haisley2010self} and \citet{gneezy2016motivated}, subjects' choices were affected by choices they made previously. Unless subjects are distracted between each choice, or choices are anonymized, this is to be expected. We are modeling a DM who wants to pool with more virtuous people. If he pools with completely different groups in rapid succession, his pretense is not very convincing. Thus, he will probably want to maintain some consistency between successive choices. To avoid consistency concerns, this paper will primarily consider individual decisions made in isolation. $c(A)$ should be interpreted as the DM's selection from menu $A$ when he does not have to make any other choices on the same domain. We return to this point in Section \ref{sec:EU}.\footnote{To preview: the model can handle bundles of decisions, but the data has to be interpreted differently. For instance, if the DM has to choose concurrently from pairs of menus, the objects of choice are of the form (item from menu 1, item from menu 2). Now, the implicit assumption is not that the DM handles each menu separately, but that he handles each pair of menus separately.}

We conclude with a note on the scope of the model. The reader may note that all the experiments above relate to morality or fairness. However, neither the formal machinery nor the interpretation of the justifiability model restricts it to these domains. It applies whenever the decision maker needs ``excuses'' for pursuing his ends (and is sometimes unable to find one). These excuses could be appeals to rationality or prudence as well as moral principles. Examples from these broader domains have not been provided here because the author is not aware of empirical papers in this area. 

\section{General model}
\subsection{Setup and notation}
\label{sec:main}
Formally, the domain $\mathcal{A}$ is a completely arbitrary set. The model is intended to apply to situations in which people care about justifying their decisions to external observers or to their ``better selves.'' Depending on the situation, a justification may be an appeal to a moral philosophy, a rational argument, or something else. The model will deliver interesting predictions whenever there is some ambiguity in what is justifiable (there are multiple acceptable choices in some situations), but not too much (some choices are outright unacceptable). Here are some examples of $\mathcal{A}$ to keep in mind: a set of plans for managing a nation's natural resources; a set of new graduates who have applied for a job; a set of funding levels (and associated bribes) for various public projects. Continuous choice variables are not a problem: there is no restriction on the cardinality of $\mathcal{A}$. 

To start, the model has two primitives. The first is the true preference $\succsim$, which is a complete and transitive relation on $\mathcal{A}$. Intuitively, it is what the DM would choose in the absence of a need to justify his decision. Formally (as the representation theorem will show), it breaks ties between different justifications. The paper does not rest on the assumption that $\succsim$ is observable: Section \ref{sec:unobserved} dispenses with $\succsim$ as a primitive and recovers it as a component of the representation. 

The second primitive is a choice correspondence $c$ that maps each non-empty finite set of alternatives to a non-empty subset. To formalize this, let $\mathcal{F}(\mathcal{A})$ be the set of non-empty finite subsets of $\mathcal{A}$. Then we have $c: \mathcal{F}(\mathcal{A}) \rightrightarrows \mathcal{F}(\mathcal{A})$ such that $c(A) \subseteq A$ for all $A \in \mathcal{F}(\mathcal{A})$.

Definition \ref{def:main} presents the most general justifiability representation. The representation is identified with a set $\mathcal{M}$ of complete, transitive and antisymmetric orders on $\mathcal{A}$.\footnote{There is no loss of generality in taking the justifications to be antisymmetric.} $\mathcal{M}$ is the set of justifications, i.e. the set of preferences that the DM considers acceptable. These preferences are not required to be continuous, so they are not guaranteed to have utility functions. However, readers who prefer to think in terms of utility functions will not lose anything by doing so. 

\begin{definition}
\label{def:main}
A justifiability representation for $(\succsim, c)$ is a nonempty set $\mathcal{M}$ of total orders (complete, transitive, antisymmetric) such that, for all $A \in \mathcal{F}(\mathcal{A})$,
\begin{align*}
&c(A) = \argmax\left(M(A), \succsim\right) \\
&\text{where } M(A) = \bigcup_{\succ_m \in \mathcal{M}}\argmax\left(A, \succ_m\right).
\end{align*}
\end{definition}

\subsection{Axioms}

Now we proceed to the axioms. The first one, Optimization, says that the DM is truly indifferent between all the items he actually selects. Although this is a standard assumption, it is possible to imagine morally-motivated DMs who violate it. For instance, a DM might feel that it is acceptable to select a selfish alternative as long as he also selects an unselfish one. This behavior is not captured by the justifiability model. In any case, Optimization has no bite when $c$ is a choice function rather than a choice correspondence. 

\begin{axiom}[Optimization]
For any $A \in \mathcal{F}(\mathcal{A})$, for any $a, b \in c(A)$: $a \sim b$. 
\end{axiom}

The second axiom, Irrelevance of Unjustifiable Alternatives (IUA), is the heart of the model. To understand it, suppose that it is unjustifiable to prefer $a$ over everything in $A$: every justifiable preference ranks at least one item in $A$ above $a$. It is obvious that $a$ will never be chosen when everything in $A$ is present. It may be less obvious that the set of justifiable alternatives---and, by extension, choice---is totally unaffected by $a$. We claim that, for any set $B$ containing $A$ and $a$, $M(B) = M(B \setminus \{a\})$. It is obvious that $M(B) \subseteq M(B \setminus \{a\})$, so we only need to show that $M(B \setminus \{a\}) \subseteq M(B)$. If $b \in M(B \setminus \{a\})$, there must be a justifiable preference that ranks $b$ above everything else in $B \setminus \{a\}$, including everything in $A$. Since every justifiable preference ranks at least one item in $A$ above $a$, this preference must rank $b$ above $a$. We conclude that there is a justifiable preference that ranks $b$ above everything else in $B$, so $b \in M(B)$. 

\begin{axiom}[Irrelevance of Unjustifiable Alternatives (IUA)]
For any $a \in \mathcal{A}$ and $A \in \mathcal{F}(\mathcal{A})$ such that $a \in A$: if $a \succsim c(A)$ and $a \notin c(A)$, then for all $B \supset A$, $c(B \setminus \{a\}) = c(B)$. 
\end{axiom}

$a \succsim c(A)$ means that, for every $b \in c(A)$, $a \succsim b$. 

\begin{example}
\label{ex:charities}
Recall the discussion of \citet{rodriguez2012self} from the introduction. As in that experiment, suppose that the DM must divide $\$30$ earned by himself and two other subjects. Suppose that subject 1 earned more than the DM, and subject 2 earned less. Let
\begin{align*}
a &= (12 \text{ to self}, 12 \text{ to }1, 6 \text{ to }2) \\
b &= (10 \text{ to self}, 14 \text{ to }1, 6 \text{ to }2) \\
d &= (10 \text{ to self}, 10 \text{ to }1, 10 \text{ to }2).
\end{align*}
Suppose that the DM's preferences are $a \succ b \succ d$: he believes in rewarding good performance, but above all wants to keep more for himself. Suppose that he chooses $a$ from $\{a, b\}$ and $\{a, d\}$, but $b$ from $\{a, b, d\}$. Intuitively, the DM can choose $a$ over $b$ because it is more equitable, and $a$ over $d$ because performance should be rewarded. But he cannot choose $a$ over both $b$ and $d$: someone who really cared about equity would choose $d$, and someone who really cared about performance would choose $b$. IUA says the the DM cannot flip from $b$ to $d$ if $a$ is removed. Moreover, if a third option $e$ is added to $b$ and $d$, the DM must make the same choice whether $a$ is available or not.
\end{example}

Theorem \ref{thm:main} is the most general representation theorem in the paper. It says that Optimization and IUA are sufficient as well as necessary for a justifiability representation. The proof proceeds in two parts. For the first part, say that $A$ ``excludes'' $a$ if $a \succsim c(A)$ and $a \notin c(A)$. We know that each acceptable preference must rank at least one item in $A$ above $a$. Thus, $\mathcal{M}$ cannot include any preference that ranks $a$ above everything in $A$. We define $\mathcal{M}$ to be the set of strict preferences on the domain that respect all these exclusion conditions. The second, and more involved, part of the proof establishes that $\mathcal{M}$ is big enough: for any $b \in c(B)$, we can find some member of $\mathcal{M}$ that ranks $b$ above everything else in $B$. This step is challenging because the exclusion conditions interact with one another, so building up a preference that satisfies them (as well as $b \succ B$) is not straightforward. 

\begin{theorem}
\label{thm:main}
$(\succsim, c)$ has a justifiability representation if and only if it satisfies IUA and Optimization. 
\end{theorem}

Since Optimization has no bite when $c$ is a choice function, IUA alone is necessary and sufficient in that case. 

\subsection{Adding dominance}
\label{sec:dominance}

In some settings, there are obvious restrictions on the preferences that a good and reasonable person might have. For instance, if the domain is a set of lotteries, a reasonable DM cannot prefer a first-order stochastically dominated lottery. If the domain is a set of payments to the DM and others, a good and reasonable DM cannot prefer a Pareto-dominated set of payments. Some more interesting, but more involved, examples are discussed at the end of this section.

First, we extend Theorem \ref{thm:main} to account for a general dominance relation. Let $\succ_D$ be a transitive and asymmetric relation on $\mathcal{A}$. $\succ_D$ is observable; it captures the analyst's existing convictions about which preferences could possibly count as ``justified.'' Definition \ref{def:D-mono} clarifies what it means to respect $\succ_D$.

\begin{definition}[Strict $D$-monotonicity]
\label{def:D-mono}
A relation $\succsim_R$ on $\mathcal{A}$ is strictly $D$-monotone if, for any $a, b \in \mathcal{A}$: $a \succ_D b$ implies $a \succ_R b$. 
\end{definition}

\begin{definition}[Monotone justifiability representation]
\label{def:D-mono_rep}
A justifiability representation $\mathcal{M}$ is $D$-monotone if each $\succ_m \in \mathcal{M}$ is strictly $D$-monotone. 
\end{definition}

The concept of exclusion was introduced in Section \ref{sec:main}. To simplify presentation of the next axiom, we now define it formally. 

\begin{definition}[Exclusion]
For any $A \in \mathcal{F}(\mathcal{A})$ and $b \notin A$, $A$ excludes $b$ if $b \succsim c(A \cup \{b\})$ and $b \notin c(A \cup \{b\})$. 
\end{definition}

Unsurprisingly, the key axiom for the augmented model is a strengthening of IUA. It says that both dominatated items and unjustifiable items are irrelevant. Formally, for any set $B$, let $S(B)$ be the set of items that are strictly dominated by something in $B$ or excluded by some subset of $B$:
\[S(B) := \{b \in B: \exists b' \in B \text{ s.t. } b' \succ_D b, \text{ or } \exists B' \subset B \text{ s.t. } B' \text{ excludes } b\}.\]
Irrelevance of Submaximal Alternatives (ISA) says that choice is unchanged when any subset of $S(B)$ is removed. 

\begin{axiom}[Irrelevance of Submaximal Alternatives (ISA)]
For any $B \in \mathcal{F}(\mathcal{A})$, for any $A \subseteq S(B)$: $c(B) = c(B \setminus A)$. 
\end{axiom}

The reader may wonder why ISA allows removal of several items, while IUA only allows removal of one item. Intuitively, this is because a variant of exclusion satisfies a nice transitivity property, which allows us to remove unjustifiable items sequentially rather than all at once. This transitivity property doesn't hold once we introduce dominance, so we aren't always able to remove submaximal items sequentially. We have to explicitly allow removing multiple items at once. 

Proposition \ref{prop:dominance} says that replacing IUA with ISA delivers a justifiability representation in which all the acceptable preferences respect dominance. The proof is along the same lines as that of Theorem \ref{thm:main}, but the construction must now keep track of dominance conditions as well as exclusion conditions. 

\begin{proposition}
\label{prop:dominance}
$(\succsim, c)$ has a $D$-monotone justifiability representation if and only if $c$ satisfies ISA and Optimization.
\end{proposition}

As promised, Example \ref{ex:dominated} shows how notions of disinterestedness or impartiality can be captured by a dominance relation. Readers less interested in extended examples may skip to Section \ref{sec:unobserved_EU} without loss of continuity. 

\begin{example} \
\label{ex:dominated}
\begin{enumerate}
\item As in \citet{gneezy2016motivated}, let the DM be an investment advisor, and $\mathcal{A}$ be a set of investments. Each investment is characterized by a distribution over payoffs $p$ and a real number $b$. The real number is the bribe the DM will receive if he recommends that investment to his client. The DM is not supposed to take bribes into account, but he can use his own judgment to determine the most attractive investment. Each justifiable preference $\succ_m$ must have $(p, b) \succ_m (p', b')$ if $p >_{FOSD} p'$, regardless of $b$ and $b'$. 
\item As in \citet{rodriguez2012self}, let $\mathcal{A}$ be a set of payments to the DM and other subjects. Each payment $p$ is a vector in $\mathbb{R}^n_+$, where the first entry $p(1)$ is the payment to the DM. Suppose that each subject $i$ is associated with a real number $a(i)$, which reflects performance in an earlier part of the experiment. Suppose the DM feels that it would be wrong to favor any particular person, himself included, but not to reward high-performing subjects. Say that a permutation of subjects $\pi$ is ability-preserving if $a(i) = a(\pi(i))$ for each $i$. Each justifiable preference $\succ_m$ must have $p \succ_m q$ if there is an ability-preserving permutation $\pi$ such that 
\[(p(1), \ldots, p(n)) > (q(\pi(1)), \ldots, q(\pi(n))).\]
To understand this, suppose that there are three subjects. Suppose that the DM and subject A performed equally well, and subject B performed worse. If the DM faces a choice between (5 to self, 0 to others) and (6 to A, 0 to others), we know he must choose the latter. ISA says more than that: the former option is totally irrelevant when the latter is present. For instance, suppose the DM chooses (6 to A, 0 to others) over the ``compromise option'' that gives 1 to everyone. ISA says he cannot flip to the compromise option when (5 to self, 0 to others) is added. More generally, adding an attractive but unjustifiable option that the DM feels compelled to forego cannot have any effect on choice. 
\end{enumerate}
\end{example}

\section{Expected-utility case}
\label{sec:EU}

\subsection{Setup and notation}

This section covers the expected-utility version of the justifiability model, in which both the true preference and the acceptable preferences have expected-utility representations, and the set of acceptable Bernoulli utilities is compact and convex. Obviously, a larger set of axioms is needed to achieve this additional structure, but the result is a tractable model suited to application. Additionally, the EU version of the model makes it easier to see where the justifiable preferences come from and how they relate two the true preference. They can easily be visualized in two or three dimensions. 

I assume that the set of prizes $Z$ is finite, although the domain $\mathcal{A} := \Delta(Z)$ is not. It is convenient to assume that there is a dominance relation $\succ_D$ on $Z$. This relation should be asymmetric and transitive, but need not be complete; it is enough to have two prizes ranked by dominance. For instance, the domain could contain two monetary prizes paid to the DM as well as prizes paid to others, or donations to various charities. A version of the representation theorem is available without this assumption, but the axioms are slightly more involved.\footnote{In the absence of dominance, the justifiable set may contain a preference exactly opposite the true preference. This case is inconvenient and needs to be handled separately.}

We will require the true preference to be strictly $D$-monotone (recall Definition \ref{def:D-mono}) and the justifiable preferences to be weakly $D$-monotone. The main justification here is convenience: weak monotonicity of the justifiable preferences falls out of the neatest set of axioms. Note that the DM will never actually \textit{choose} a strictly dominated alternative. If the dominated alternative is justifiable, the dominating alternative will be too, and the DM will choose the latter but not the former. 

\begin{definition}[Weak $D$-monotonicity]
A relation $\succsim_R$ on $Z$ is weakly $D$-monotone if, for any $a, b \in \mathcal{A}$: $a \succ_D b$ implies $a \succsim_R b$. A function $u$ on $Z$ is weakly (strictly) $D$-monotone if $a \succ_D b$ implies $u(a) \geq (>) u(b)$.
\end{definition}

We are now ready to formally define an EU justifiability representation. 

\begin{definition}[EU Justifiability Representation]
An EU justifiability representation consists of a strictly $D$-monotone Bernoulli utility $u$ and a compact, convex set $\mathcal{M}^{EU}$ of weakly $D$-monotone Bernoulli utilities such that $p \mapsto \mathbb{E}u(p)$ represents $\succsim$, and
\begin{align*}
&c(A) = \argmax_{a \in M(A)} \mathbb{E}_a u \\
&\text{where } M(A) := \bigcup_{m \in \mathcal{M}^{EU}} \argmax_{a \in A} \mathbb{E}_a m.
\end{align*}
\end{definition}

One more piece of notation is needed before proceeding to the axioms. For any $p \in \Delta(Z)$, let $B(p)$ be the set of items that are no better than $p$, but defeat $p$ in pairwise comparison. Formally, 
\begin{equation}
\label{eq:Bp}
B(p) := \{q \in \Delta(Z): p \succsim q \text{ and } \{q\} = c(\{p, q\})\}.
\end{equation}
Let $NB(p)$ be the complement of $p$ in $\{q \in \Delta(Z): p \succsim q\}$:
\begin{equation}
NB(p) := \{q \in \Delta(Z): p \succsim q \text{ and } p \in c(\{p, q\})\}.
\end{equation}
$B(p)$ is important because its common boundary with $NB(p)$ pins down the set of acceptable preferences. It is like a ``sufficient statistic'' for the representation. 

\subsection{Axioms}

Each axiom in this section has two parts. For all axioms but Convexity, the first part is a condition on the true preference, and the second part is a condition on choice behavior that ultimately translates into a condition on justifiable preferences. Since the conditions for a given preference to have a FOSD-monotone EU representation are well known, the first part is completely standard. 

\begin{axiom}[Independence] \
\begin{enumerate}
\item For all $p, q, r \in \Delta(Z)$ and $\alpha \in (0, 1)$, $p \succsim q$ implies $\alpha p + (1-\alpha) r \succsim \alpha q + (1-\alpha) r$. 
\item For any $A \in \mathcal{F}(\Delta(Z))$ and $p \in \Delta(Z)$, 
\[c(\alpha A + (1-\alpha)\{p\}) = \alpha c(A) + (1-\alpha)\{p\}.\]
\end{enumerate}
\end{axiom}

The second part of Independence says that the DM's preference does not flip when every option he faces is mixed with a fixed lottery in a fixed proportion. The necessity of this axiom for an EU justifiability representation is obvious. If $p$ is the best item in $A$ that is top-ranked by some acceptable preference, then mixing everything in $A$ (including $p$) with $q$ will not change this. Formally, this axiom ensures that $B(p)$ is a convex cone. The acceptable preferences are among the supporting hyperplanes of this cone. 

\begin{axiom}[Continuity] \
\begin{enumerate}
\item For any $p \in \Delta(Z)$, $\{q \in \Delta(Z): q \succsim p\}$ and $\{q \in \Delta(Z): q \precsim p\}$ are closed.
\item For any $p \in \Delta(Z)$, $NB(p)$ is closed. 
\end{enumerate}
\end{axiom}

There is nothing unusual about the continuity axiom. If $NB(p)$ fails to be closed, the set of justifiable preferences constructed in the proof will be slightly too permissive. Some preferences will be indifferent between pairs of items that they should strictly rank. 

For the monotonicity condition, we need to extend first-order stochastic dominance (FOSD) to allow for an incomplete ranking over prizes. Notice that the definition provided here reduces to the usual one when $\succ_D$ is complete. 

\begin{definition}[FOSD]
$p >_{FOSD} q$ if $p = \sum_i \alpha_i \delta_{p_i}$, $q = \sum_i \alpha_i \delta_{q_i}$, $p_i \succ_D q_i$ or $p_i = q_i$ for all $i$, and $p_i \succ_D q_i$ for some $i$.
\end{definition}

\begin{axiom}[Monotonicity] \
\begin{enumerate}
\item $\succsim$ is strictly FOSD-monotone. 
\item For any $p, q \in \Delta(Z)$ such that $p >_{FOSD} q$ and any $A \in \mathcal{F}(\Delta(Z))$ containing $p, q$, $c(A) = c(A \setminus \{q\})$. 
\end{enumerate}
\end{axiom}

The second part of Monotonicity is really Irrelevance of Dominated Alternatives (IDA). It says that any lottery is irrelevant in the presence of a lottery that strictly dominates it. This type of condition should be familiar from Section \ref{sec:dominance}. 

The final axiom, Convexity, is probably the least familiar. The first part says that a menu $A$ excludes a lottery $p$ only if there is some mixture of lotteries in $A$ that is weakly worse than $p$ but beats $p$ in pairwise comparison. The second part is a partial converse of the first. If there is some mixture of lotteries in $A$ that beats $p$ in pairwise comparison, then $p$ cannot be chosen from $\{p\} \cup A$. This might be because $p$ is unjustifiable when $A$ is present, or simply because something better than $p$ is available. 

The word ``mixture'' is important. Convexity does not say that $A$ excludes $p$ only if some member of $A$ is weakly worse than $p$ but beats $p$ in pairwise comparison. This is typically not the case. Unless the set of justifiable preferences is a singleton, we can find a lottery $p$ and menu $A$ such that $p$ beats each member of $A$, but $A$ beats $p$. Indeed, the inconsistency between choices in binary menus and choices in larger menus is what makes the justifiability model interesting. 

The name Convexity arises because necessity of the first part follows from convexity of the justifiable set. To see what role Convexity plays in the proof of sufficiency, suppose that $p$ is the DM's favorite item in a given menu.\footnote{This is not the only case in which Convexity is applied, but it is the most straightforward one.} The first part of the axiom implies that the DM will choose $p$ provided $B(p)$ can be separated from $A$ by a single, appropriately chosen hyperplane. This hyperplane will turn out to be the indifference curve of a justifiable preference. The second part of the axiom implies that the DM will not choose $p$ if it cannot be separated from $A$ in this way. Thus, the DM does not need to appeal to nonlinear preferences to justify any of his choices. 

\begin{axiom}[Convexity]
For any $A \in \mathcal{F}(\Delta(Z))$ and $p \notin A$:
\begin{enumerate}
    \item If $A$ excludes $p$, then $\text{co}(A) \cap B(p) \neq \emptyset$.
    \item If $\text{co}(A) \cap B(p) \neq \emptyset$, then $p \notin c(A \cup \{p\})$. 
\end{enumerate}
\end{axiom}

Theorem \ref{thm:EU} is the representation theorem for the EU case. The proof is quite different from that of Theorem \ref{thm:main}; it uses mostly geometric rather than order-theoretic arguments. First, Independence and Continuity are used to show that $B(p)$ is always a convex cone, open in $\{q \in \Delta(Z): p \succsim q\}$. Monotonicity is used to establish the relationship between $B(p)$ and the true preference. $B(p)$ is used to identify a candidate set of justifiable preferences. Each candidate preference has an indifference curve that is a supporting hyperplane of $B(p)$. Finally, Convexity and the axioms from Theorem \ref{thm:main} are used to show that the candidate set is neither too large nor too small. No candidate preference would allow the DM to justify a choice he would have liked to make, but did not; and for each choice the DM actually makes, some candidate preference justifies it. 

\begin{theorem}
\label{thm:EU}
Suppose that $Z$ is finite and that at least two prizes are dominance-ranked. The following are equivalent:
\begin{enumerate}
\item $(\succsim, c)$ satisfies IUA, Optimization, Independence, Continuity, Monotonicity and Convexity. 
\item $(\succsim, c)$ has an EU justifiability representation. 
\end{enumerate}
\end{theorem}

The EU justifiability model is subject to the usual uniqueness issues for EU representations: both the true preference and the justifiable preferences can be arbitrarily (and independently) shifted and scaled. There is also a more subtle uniqueness issue: some EU preferences that are \textit{not} positive affine transformations of other preferences can be added to or removed from the set of justifiable preferences without changing anything. Intuitively, this is because these preferences are too far from the true preference to be of much use. Any choice they justify is also justified by some preference that is closer to $\succsim$. 

To get around all the uniqueness issues at once, we restrict attention to the maximal and minimal sets of justifiable preferences. For the maximal set, we include all the preferences not ruled out by the DM's behavior, or dominance; for the minimal set, we include only those preferences needed to explain the DM's behavior. In both cases, focusing on preferences rather than utility functions gets rid of shifting and scaling. Formally, fix any EU justifiability representation $(u, \mathcal{M}^{EU})$ of $(\succsim, c)$. Let $\mathcal{M}^{EU}_\succ$ be the set of preferences that have EU representations with Bernoulli utilities in $\mathcal{M}^{EU}$:
\[\mathcal{M}^{EU}_\succ := \left\{P \subset \mathcal{A}^2: \exists m \in \mathcal{M}^{EU} \text{ s.t. } p \mapsto \mathbb{E}_p m \text{ represents } P\right\}.\]
We call $\mathcal{M}^{EU}_\succ$ a set of justifiable EU preferences for $(\succsim, c)$. 

\begin{corollary}
\label{cor:EU_unique}
If $(\succsim, c)$ has an EU justifiability representation, it has a unique maximal and a unique minimal set of justifiable EU preferences.
\end{corollary}

We conclude with a note on multiple decisions made simultaneously, or in immediate succession. In many experiments, subjects are asked to make multiple decisions, and told that one will be randomly selected to be implemented. As observed by \citet{holt1986preference} and \citet{karni1987preference}, this bundling will not affect choice when subjects have EU preferences. This is not true in the EU version of the justifiability model because subjects are not maximizing a single EU preference, but choosing according to the justifiable EU preference that gets them closest to their true preference. When they must make multiple decisions simultaneously, they may feel compelled to use \textit{the same} justifiable preference for all decisions. Thus, it is important to separate decisions as much as possible when collecting data for the model.\footnote{This could be done by introducing a distraction phase between questions. Reducing pressure to appear consistent may also help, e.g. by assuring subjects that no one but the experimenter will observe choices not selected to be implemented.}

This is not to say that the justifiability model is silent on simultaneous decisions. Once the DM's true preference and set of justifiable preferences have been identified, the model makes precise predictions about the DM's behavior. To see why, suppose the DM must choose simultaneously from menu $A$ and menu $B$, and that each choice will be implemented with equal probability. He should choose lottery $a$ from $A$ and lottery $b$ from $B$ in the double-decision case if and only if he would choose lottery $(1/2)a + (1/2)b$ from menu $(1/2)A + (1/2)B$ in the single-decision case. 

\subsection{Comparative Statics}

We might like to determine whether one DM has a more expansive or restrictive view of morality (or some other constraint) than another. These two DMs may be different people whose true preferences happen to coincide, or they may be the same person in different situations. Intuitively, DM's true preference should not change with the situation, but the standard for a ``justification'' may. Thus, the DM may have access to a larger or smaller set of justifications in different situations. We return to this idea in Section \ref{sec:double}, once we have dealt with the unobserved-true-preference case. For now, we remain agnostic about the precise interpretation of  $(\succsim, c_1)$ and $(\succsim, c_2)$.

Since the proof of Theorem \ref{thm:EU} uses $B(p)$ to obtain a set of justifiable EU preferences, we may expect a connection between the size of $B(p)$ and the strictness of the DM's morals. There is indeed a connection. If $B_1(p) and B_2(p)$ are nested, then we can find representations in which the sets of justifiable Bernoulli utilities are also nested. If $B_2(p)$ is the smaller set, then DM 2 can appeal to the larger set of justifiable utilities. This is convenient because we can figure out each DM's view of morality just by looking at his choices on a small set of menus: binary menus that contain a fixed item $p$. If DM 1 chooses a weakly inferior item $q$ over $p$ whenever DM 2 does, we can assume that DM 1's morals are at least as strict as DM 2's. 

\begin{corollary}
\label{cor:comp_stat_EU}
Suppose that $(\succsim, c_1)$ and $(\succsim, c_2)$ have EU justifiability representations. The following are equivalent:
\begin{enumerate}
\item $B_1(p) \supset B_2(p)$ for some $p \in \text{int}(\Delta(Z))$.
\item The maximal set of justifiable EU preferences for $(\succsim, c_1)$ is strictly smaller than the maximal set of justifiable EU preferences for $(\succsim, c_2)$.
\end{enumerate} 
\end{corollary}

The result is not true if ``maximal'' is replaced with ``minimal.'' This is because there may be justifications that the more liberal DM never feels the need to use. He may agree that these are acceptable justifications, but he doesn't appeal to them because something better is always available. Thus, his minimal set of justifiable EU preferences doesn't include them.

\section{What if the true preference is unobserved?}
\label{sec:unobserved}

Since the axioms for both the general and EU versions of the model make extensive use of the true preference, one might think that Theorems \ref{thm:main} and \ref{thm:main} are useless when the true preference are unobservable. This is not the case. In both versions of the model, choice behavior places substantial restrictions on the true preference, and sometimes, it even fully identifies the true preference. Where it does not, it still identifies a set of candidates. Any of these candidates can be used to verify the axioms. There is no need to check the axioms for each candidate in turn: there is a justifiability representation if and only if the axioms are satisfied conditional on an arbitrarily chosen candidate. This section explains how to obtain the candidate(s) from choice behavior.

Throughout, I use ``$c$ has a justifiability representation'' as shorthand for ``there is some preference $\succsim$ for which $(\succsim, c)$ has a justifiability representation.''

\subsection{Expected-utility case}
\label{sec:unobserved_EU}

We start with the EU version of the model because it is more straightforward. Typically, there will be a unique candidate for the true preference. The ``untypical'' case corresponds to an EU justifiability representation in which $\mathcal{M}^{EU}_\succ$ is a singleton, and the true preference is used only for tiebreaking. This case can be spotted in the data because $\{q \in \Delta(Z): p \in c(\{p, q\})\}$ is a half-space if and only if $\mathcal{M}^{EU}_\succ$ is a singleton. 

In the typical case, the unique candidate for the true preference is obtained in two steps. First, determine the true indifference curves. This is easy because the DM only chooses both items from a binary set if he is indifferent between them. Once the indifference curves are pinned down, there are only two candidates, distinguished by the direction of preference. The correct candidate is the one that makes $B(p)$ convex. Proposition \ref{prop:unobserved_EU} states this formally. 

\begin{proposition}
\label{prop:unobserved_EU}
Fix $c$ and $p \in \text{int}(\Delta(Z))$. Suppose that $\{q \in \Delta(Z): p \in c(\{p, q\})\}$ is not a half-space (restricted to the simplex). 
\begin{enumerate}
\item There is at most one preference $\succsim$ for which $(\succsim, c)$ has an EU justifiability representation.
\item If there is such a preference, it is the unique EU preference such that $\{p, q\} = c(\{p, q\})$ implies $p \sim q$, and $B(p)$ is convex.
\end{enumerate}
\end{proposition}

The case excluded by Proposition \ref{prop:unobserved_EU} is equally simple, but it is relegated to the Appendix for brevity. A result similar to Proposition \ref{prop:unobserved_EU}, minus uniqueness, is still available. Rather than reading off a unique candidate for the true preference, one can read off a set of candidates. It is only necessary to check one of them. Either they all work ($(\succsim, c)$ has an EU justifiability representation for each candidate $\succsim$), or none of them does (there is no preference $\succsim$ for which $(\succsim, c)$ has an EU justifiability representation).

\subsection{General case}

Now we return to the general case of the model. For simplicity, we take $c$ to be a choice function rather than a choice correspondence. We are looking for one or more strict preferences $\succ$ such that $(\succ, c)$ has a justifiability representation. 

The analysis in this section proceeds in two steps. In Section \ref{sec:ax_1}, we show that a straightforward modification of Theorem \ref{thm:main} delivers a representation result for the unknown-$\succ$ case. The corresponding axiom is intuitive and easy to check. It also provides insight into the DM's true preference: it tells the researcher exactly which preferences are consistent with the DM's behavior, and which are not. However, it does not provide much (direct) insight into the preferences the DM considers justifiable. Section \ref{sec:ax_2} corrects that deficiency. Through an alternative axiomatization, it shows how to identify the preferences that the DM may consider justifiable, and how to rule out the rest. Conveniently, the constraints on the justifiable preferences come from simple, easily recognized patterns of choice. 

\subsubsection{First Axiomatization}
\label{sec:ax_1}

This section builds directly on Theorem \ref{thm:main}. Recall that IUA is a necessary and sufficient condition for a justifiability representation when $c$ is a choice function. It says that $a$ is irrelevant in the presence of $A \setminus \{a\}$ if $a$ is better than $c(A)$ according to the true preference. We can flip IUA on its head to derive conditions on the true preference. To see how this works, suppose that choice from $B$ changes when item $a$ is removed. Clearly, $a$ is not irrelevant in the presence of $B$ or any of its subsets. Take any subset of $B$ that contains $a$. By IUA, the DM cannot strictly prefer $a$ to the item he selects from this subset. If the DM does not select $a$, he must strictly prefer the item he does select to $a$. This is a clear restriction on $\succ$. We can apply this procedure until we run out of data. If the resulting restrictions form a cycle, no strict preference can satisfy them, so there is no representation. But if the restrictions don't form a cycle, $c$ will satisfy IUA conditional on any strict preference that obeys the restrictions, so Theorem \ref{thm:main} delivers a representation. This is exactly what Corollary \ref{cor:unknown_true_pref} says. 

\begin{definition}[Revealed Preference (P)]
If $c(B) \neq c(B \cup \{a\})$, then for any $A \subseteq B$ such that $a \neq c(A \cup \{a\})$, $c(A \cup \{a\})$ is revealed preferred to $a$ ($c(A \cup \{a\}) \; P \; a$). 
\end{definition}

\begin{axiom}[Acyclicity]
$P$ is acyclic. 
\end{axiom}

\begin{corollary}
\label{cor:unknown_true_pref}
$c$ satisfies Acyclicity if and only if it has a justifiability representation. Moreover, the true preference in every justifiability representation for $c$ extends $P$.
\end{corollary}

Corollary \ref{cor:unknown_true_pref} is helpful not just because it delivers a representation, but because it tells us exactly what the true preference (in any representation) must look like. Sometimes, it pins down a unique preference. Example \ref{ex:disabled} illustrates.

\begin{example}
\label{ex:disabled}
Recall the discussion of \citet{snyder1979avoidance} at the beginning of Section \ref{sec:empirics}. We will use the same example (and notation) here, with one modification. \citeauthor{snyder1979avoidance} restricted attention to binary menus, but menus with more than two alternatives are important for identification in the justifiability model. Thus, we need to specify choice from the grand set $\{a_1, a_2, b_1\}$ as well as all the pairs:
\[c(\{a_1, a_2\}) = a_1 \quad c(\{a_1, b_1\}) = b_1 \quad c(\{a_2, b_1\}) = c(\{a_1, a_2, b_1\}) = b_1.\]
Notice that choice from the grand set changes when $b_1$ or $a_2$ is removed. Since $a_2$ matters and $a_1 = c(\{a_1, a_2\})$, we must have $a_1 \succ a_2$. Since $b_1$ matters and $c(\{a_2, b_1\}) = a_2$, we must have $a_2 \succ b_1$. Putting these two restrictions together, we get $a_1 \succ a_2 \succ b_1$.\footnote{The model of \citet{cherepanov2013rationalization} delivers only $a_1 \succ a_2$. Thus, their main model does not rule out $b_1 \succ a_1 \succ a_2$, in which the DM prefers sitting with the disabled stranger to watching either movie alone.} The DM prefers movie A to movie B, but above all prefers to avoid the stranger. This is precisely the preference we conjectured in Section \ref{sec:empirics}. Now we have shown that the conjecture is, in fact, the true preference. 
\end{example}

Conditional on any true preference $\succ$, we can work out the set of justifiable preferences. Recall that $A$ ``excludes'' $a$ if $a \succ c(\{a\} \cup A)$. This translates into an ``exclusion condition'' on the set of justifiable preferences: no justifiable preference can rank $a$ above $A$. As shown in the proof of Theorem \ref{thm:main}, we can take the set of justifiable preferences to be precisely the set of preferences that satisfy all these exclusion conditions. 

Still, this process leaves something to be desired. It would be better to learn about the set of justifiable preferences just by looking at the data, not by constructing a set of true preferences and then working out the constraints associated with each one. The next section explains how to do that. The reader is forewarned that the second axiomatization is somewhat more involved than the first one. It should still be of interest to the empirically inclined reader, though.

\subsubsection{Second Axiomatization}
\label{sec:ax_2}

Two patterns of choice behavior are key to understanding the set of justifiable preferences. We define these patterns, explain their implications for the true preference and the justifiable preferences, and leverage them to obtain a representation theorem for the unknown-true-preference case. 

The first key pattern is a three-element cycle. As suggested in Example \ref{ex:disabled}, the true preference on any three-element cycle is uniquely pinned down. The item chosen from the full three-element set is middle-ranked, and the item that beats it is top-ranked. Further restrictions on the true preference may be obtained by chaining together multiple cycles. For instance: if one cycle reveals $a$ to be better than $b$, and another reveals $b$ to be better than $d$, $a$ is certainly better than $d$. 

\begin{definition}[Cycle/Chain]
\label{def:cycle}
$(a_1, a_2, a_3)$ is a cycle if
\[c(\{a_1, a_2\}) = a_1 \quad c(\{a_1, a_2, a_3\}) = a_2 \quad c(\{a_1, a_3\}) = a_3.\]
For $k \geq 3$, $(a_1, \ldots, a_k)$ is a chain if for each $i \in \{2, \ldots, k-1\}$, $(a_{i-1}, a_i, a_{i+1})$ is a cycle and/or both $(a_{i-2}, a_{i-1}, a_i)$ and $(a_i, a_{i+1}, a_{i+2})$ are cycles. 
\end{definition}

Cycles tell us about the justifiable preferences as well as the true preference. The best item in any cycle is pairwise-defeated by the worst item, so the best item must be unjustifiable when the latter is present. Again, we may be able to obtain more information by chaining together multiple cycles. If $a$ is revealed better than $b$ and $b$ is revealed better than $d$, but the DM chooses $d$ from $\{a, d\}$, $a$ must be unjustifiable when $d$ is present. 

The second key pattern of choice is an almost-WARP set. (The reasons for the name will soon become clear.) Fix some set $A$, and suppose that choice satisfies WARP on all its proper subsets. If $|A| = 3$, suppose further that pairwise choice is not cyclic. (Pairwise choice can't be cyclic if $|A| > 3$, and we have already dealt with 3-element cycles.) Then, there is a unique preference on $A$ that is maximized by choice from each proper subset. This preference is pinned down by pairwise choice. If choice on $A$ violates WARP, so $c(A)$ is pairwise-defeated by some other item in $A$, the true preference on $A$ is uniquely pinned down. In fact, it is the preference given by pairwise choice. (This may not be immediately obvious, but it is straightforward to prove.) 

\begin{definition}[Almost-WARP set]
Suppose that $A$ is not a cycle. $A$ is an almost-WARP set if choice violates WARP on $\mathcal{F}(A)$, but satisfies WARP on $\mathcal{F}(A) \setminus A$. 
\end{definition}

Like cycles, almost-WARP sets tell us about the justifiable preferences as well as the true preference. Since the best item is not chosen from the full set, it must be unjustifiable in the presence of the other items.

The implications of cycles and almost-WARP sets for justifiable preferences are summed up in Definition \ref{def:rev_exc}. 

\begin{definition}[Revealed exclusion]
\label{def:rev_exc}
$a$ is revealed excluded by $B$ if:
\begin{enumerate}
\item For $|B| > 1$: $B \cup \{a\}$ is an almost-WARP set, and $a$ pairwise-defeats $c(B \cup \{a\})$. 
\item For $B = \{b\}$: $b = c(\{a, b\})$, and $a$ comes before $b$ in a chain.
\end{enumerate}
\end{definition}

Since three-element cycles and almost-WARP sets are easy to spot, so is revealed exclusion. One may wonder whether more restrictions on the justifiable preferences could be obtained from more complicated patterns of choice. The answer is no: cycles and almost-WARP sets tell us all we could hope to know about the preferences the DM considers justifiable. Every preference consistent with revealed exclusion appears in some representation. (In fact, there is a representation in which the set of justifiable preferences is \textit{precisely} the set of preferences consistent with revealed exclusion. We return to this point after the next representation result.) Proposition \ref{prop:rev_exc} summarizes. 

\begin{proposition}
\label{prop:rev_exc}
Suppose $c$ has a justifiability representation. For any $a \notin B$, the following are equivalent:
\begin{enumerate}
\item $a$ is revealed excluded by a subset of $B$. 
\item No justifiable preference in any representation ranks $a$ above $B$.
\end{enumerate}
\end{proposition}

The next axiom is the analogue of IUA for the unknown-true-preference case. It says that an item $a$ is irrelevant in the presence of a set $B$ if $a$ is revealed excluded by $B$. 

\begin{axiom}[Irrelevance of Excluded Alternatives (IEA)]
If each $a \in A \subset B$ is revealed excluded by a subset of $B$, then $c(B) = c(B \setminus A)$. 
\end{axiom}

IEA has no bite if nothing is revealed excluded, so the reader may wonder how often cycles and almost-WARP sets actually arise. The answer is reassuring: unless choice satisfies WARP (in which case a standard preference-maximization model is perfectly adequate), there will be at least one cycle or almost-WARP set, so at least one item will be revealed excluded. This will provide an opportunity to falsify the model. 

Before moving to the representation result, it is worth noting one implication of IEA. Suppose that $A$ is an almost-WARP set, so there is a unique preference maximized by choice on each of its proper subsets. We can index the items in $A$ from best to worst according to this preference: $a_1 \succ \cdots \succ a_n$. (As noted above, the true preference agrees with the WARP-implied preference, hence the notation.) Since choice on $A$ violates WARP, we can't have $c(A) = a_1$. It turns out we can only have $c(A) = a_2$---the item chosen from $A$ is the DM's second-favorite item.\footnote{Suppose $c(A) = a_3$. Then, $a_2$ is revealed excluded by $A \setminus \{a_2\}$. By IEA, we can remove $a_2$ from $A$ without changing choice. But we know that $c(A \setminus \{a_2\}) = a_1 \neq c(A)$, so we have a contradiction. This argument generalizes to $i > 3$.} This means the DM's choice cannot deteriorate too quickly as we expand the choice set. He can move from always choosing his favorite item to choosing his second-favorite, but not to his third-favorite or worse. This property requires justifications to be preferences: the model of \citet{cherepanov2013rationalization} does not share it. 

Theorem \ref{thm:unknown_true_pref} is the promised representation result. It says that IEA is necessary and sufficient for a representation. 

\begin{theorem}
\label{thm:unknown_true_pref}
$c$ satisfies IEA if and only if $c$ has a justifiability representation. 
\end{theorem}

The proof of Theorem \ref{thm:unknown_true_pref} constructs a particular justifiability representation for $c$, which we call the ``canonical representation.'' This representation is notable because the set of justifiable preferences is precisely the set of preferences consistent with revealed exclusion. Thus, the set of justifiable preferences is maximal: it includes the set of justifiable preferences from every other representation. Maximality of justifications corresponds to minimality of constraints: the more justifications are available to the DM, the fewer constraints he faces in making his decision. Therefore, the canonical representation for $c$ is the most parsimonious model of constrained decision-making that explains $c$. 

The true preference in the canonical representation is easily constructed, too. First, impose $a \succ b \succ d$ whenever $(a, b, d)$ is a cycle, and take the transitive closure. Second, if $a$ and $b$ have not yet been ranked, impose $a \succ b$ if $a = c(\{a, b\})$, and $b \succ a$ otherwise. Although this may not be the only preference consistent with behavior, it is the only preference consistent with the \textit{maximal} set of justifiable preferences.\footnote{\citet{cherepanov2013rationalization} agree that the most desirable representation is the one that imposes fewest constraints on the DM. In their model, as in this one, the true preference corresponding to the minimal-constraint representation matches pairwise choice when pairwise choice is acyclic. But \citeauthor{cherepanov2013rationalization} do not provide general results about the true preference when the data exhibits cycles. Since cycles will occur whenever one item is unjustifiable in the presence of another, it is undesirable to rule them out.} 
Corollary \ref{cor:maximal} summarizes the properties of the canonical representation. 

\begin{definition}[Canonical representation]
Representation $(\succ^*, \mathcal{M}^*)$ is canonical if (1) $\mathcal{M}^*$ is the set of preferences consistent with revealed exclusion, and (2) $\succ^*$ is the preference that has $a \succ b$ whenever $a$ comes before $b$ in a chain, and that agrees with pairwise choice on pairs not connected by any chain. 
\end{definition}

\begin{corollary}
\label{cor:maximal}
Suppose $c$ has a justifiability representation $(\succ, \mathcal{M})$. Then, it has a unique canonical representation $(\succ^*, \mathcal{M}^*)$, and $\mathcal{M}^* \supseteq \mathcal{M}$. 
\end{corollary}

\subsection{Comparing settings}
\label{sec:double}

Sections \ref{sec:ax_1} and \ref{sec:ax_2} cover identification in a single, static setting. Intuitively, we should be able to learn more about the true preference and/or set of justifiable preferences by varying the pressure to find a good justification. For instance, we would expect the DM's choices to be closer to his true preference when he chooses anonymously than when he must announce his choice to some ethically conscious peers. We do not argue that the DM is wholly unconstrained in the low-pressure setting, but that he is less constrained in the low-pressure setting than the high-pressure setting. 

To formalize this, let $c_L$ be the choice function corresponding to the low-pressure setting, and $c_H$ to be the choice function corresponding to the high-pressure setting. We are now looking for a pair of representations with the same true preference and nested sets of justifiable preferences. We accomplish this by building on Theorem \ref{thm:unknown_true_pref}. First, we ensure that $c_L$ has a justifiability representation by requiring it to satisfy IEA. Second, we impose consistency between $c_H$ and $c_L$. This consistency condition, IREA, is essentially a stronger version of IEA. It says that any item revealed excluded in the low-pressure setting is irrelevant in the high-pressure setting. This is clearly necessary for the set of justifications to be smaller in the high-pressure case. IREA also says that anything the DM chose in the low-pressure setting, but not in the high-pressure setting, is irrelevant in the high-pressure setting. Intuitively, this is because choice changes as the pressure to find a justification increases only if the item that was originally selected no longer counts as justified.

\begin{definition}[Replacement]
$a$ is replaced in $A$ if $a = c_L(A) \neq c_H(A)$.
\end{definition}

\begin{axiom}[Irrelevance of Replaced or Excluded Alternatives (IREA)]
If each $a \in A \subset B$ is revealed excluded in $L$ by, or replaced in, a subset of $B$, then $c_H(B) = c_H(B \setminus A)$. 
\end{axiom}

\begin{example}
Recall Example \ref{ex:disabled}. Augment this example to have two treatments. In the low-pressure setting $L$, the disabled stranger does not observe the DM unless the DM chooses the room in which the stranger is sitting. In the high-pressure setting $H$, the stranger observes the DM choosing between the rooms. Suppose that the DM's choices in $L$ are given in Example \ref{ex:disabled}. Suppose further than $b_1 = c_H(\{a_1, a_2, b_1\})$. IREA implies that $b_1$ is chosen in $H$ whenever it is available, including in $\{b_1, a_2\}$. Intuitively, 
\[b_1 = c_H(\{a_1, a_2, b_1\}) \neq c_L(\{a_1, a_2, b_1\}) = a_2\] 
reveals that the DM now feels unable to choose $a_2$ over both $a_1$ and $b_1$. Since he already felt unable to choose $a_1$ over $b_1$, he must now feel compelled to choose $b_1$ over both $a_1$ and $a_2$.
\end{example}

Proposition \ref{prop:double} is the representation result for the two-setting case. At the expense of additional notation, it could easily be extended to more than two settings. 

\begin{proposition}
\label{prop:double}
$c_L$ and $c_H$ have justifiability representations $(\succ, \mathcal{M}^L)$ and $(\succ, \mathcal{M}^H)$ such that $\mathcal{M}^H \subseteq \mathcal{M}^L$ if and only if $(c_L, c_H)$ satisfies IREA and $c_L$ satisfies IEA. 
\end{proposition}

Proposition \ref{prop:double} also helps to interpret the known-true-preference case in Section \ref{sec:main}. It may seem mysterious for the analyst to observe a component of the representation. Corollary \ref{cor:WARP_for_cL} explains what is really going on. Rather than directly observing the true preference, the analyst observes choice behavior in a low-pressure situation. Provided this behavior satisfies WARP, she identifies the WARP-implied preference with the true preference. Corollary \ref{cor:WARP_for_cL} says there is little harm in this: if there is any justifiability representation, there is one in which the true preference is the WARP-implied preference. 

\begin{corollary}
\label{cor:WARP_for_cL}
Suppose that $c_L$ satisfies WARP, so the restriction of $c_L$ to binary menus pins down a unique preference $\succ$. $c_H$ has a justifiability representation if and only if it satisfies IUA conditional on $\succ$. 
\end{corollary}

\section{Existing Models}
\label{sec:theory}
Several existing models are formally related to the justifiability model, but bear very different interpretations. \citet{aizerman1981general} (summarized in \citet{moulin1985choice}) present and characterize the following model:
\[c(A) = \bigcup_{\succ_m \in \mathcal{M}}\argmax(A, \succ_m).\]
In the justifiability model, the expression on the right-hand-side is simply the justifiable set $M(A)$. $M(A)$ coincides with $c(A)$ when the true preference is indifferent between everything in the domain. Despite this tight connection, the behavioral characterization of this model is very different from that of the justifiability model. Intuitively, a DM who maximizes a non-trivial preference over the justifiable set will only reveal a small subset of the justifiable set to the analyst. This makes it harder to tell what the set of justifiable preferences looks like, particularly when the true preference is unknown.

\citet{kalai2002rationalizing} study a model in which $c(A)$ is a selection from the justifiable set. Any selection will do---choice is not required to maximize any preference over the justifiable set. In sharp contrast to the justifiability model, this model has no empirical content. The only way to ``reject'' the model is to show that explaining choice requires an implausibly large number of rationales. This makes sense when maximization of a single, stable preference is taken as the benchmark for rationality. When fairly reasonable people make choices that can only be explained with an absurdly large number of rationales, we should probably look for a different model. The justifiability model requires a completely different interpretation. The DM is required to \textit{have} a single, stable preference. He sometimes fails to maximize it because he believes that certain aspects of his preference are unjustifiable. The larger the set of justifiable preferences, the less frequently this will happen. Thus, a DM whose behavior is consistent with a large set of justifiable preferences should be seen as ``relatively unconstrained,'' not ``relatively irrational.'' A large set of justifiable preferences is certainly not a reason to reject the model.

\citet{gul2005revealed}, building on \citet{strotz1955myopia}, study a model of changing tastes. Here, the DM is effectively limited to plans that he will be willing to carry out in future. Ties can be broken in accordance with the current preference over consumption streams. Formally, the two-stage version of their model is a justifiability model with a single justification (the future preference). \citeauthor{gul2005revealed} work with choice over menus rather than choice from menus. This matters because choice over singleton menus can be identified with the tiebreaker, so the tiebreaker is effectively observable.

\citet{manzini2007sequentially} also model a two-stage decision process, which they call sequential rationalizability. In their model, one ``preference'' restricts the set of alternatives, and then another breaks the ties. The word ``preference'' is potentially misleading here because \citeauthor{manzini2007sequentially} do not require the components of the representation to be complete or transitive. They are simply binary relations. The justifiability model with a single justification can be seen as a strengthening of sequential rationalizability. The properties of the model change significantly when more justifications are added, though. \citeauthor{manzini2007sequentially} show that their model requires an item that pairwise-beats all other items in the menu to be chosen from that menu. This is not the case in the justifiability model because the DM might use \textit{different} justifications to select his favorite item from each binary menu. There is no guarantee that any justification will allow him to select that item from the full menu. 

\citet{cherepanov2013rationalization} (CFS) is the paper most similar to this one, in both form and interpretation. In their model, the DM maximizes his true preference over the subset of items that he can rationalize. He can rationalize choosing $a$ over $B$ if, for each $b$ in $B$, he has a rationale for choosing $a$ over $b$. The key difference is that rationales are not required to have any particular structure. They are not required to be complete or transitive, and their properties (e.g. monotonicity, independence) are not explored. Although this may seem a technical point, it makes application and interpretation difficult. Recall from Section \ref{sec:unobserved} that two patterns of choice behavior---cycles and almost-WARP sets---pin down the DM's true preference and the constraints he faces. The rationalizability model says less than the justifiability model about both patterns of choice. As discussed in Example \ref{ex:disabled}, it cannot pin down the true preference (and, by extension, the constraint) when pairwise choice is cyclic. It cannot pin down the true preference on almost-WARP sets either; it can only say that the item selected from the full set is worse than any item that defeats it.

CFS address these problems by assuming that the analyst knows some of the DM's rationales. These rationales cannot be inferred from choice behavior; they represent the analyst's ``outside knowledge'' about the DM. By judicious choice of rationales, CFS can explain behavior in Example \ref{ex:disabled}. When choice is acyclic, CFS show that maximizing the set of rationales pins down a unique representation. This result bears some similarities to Corollary \ref{cor:maximal}. However, Corollary \ref{cor:maximal} does not require acyclic choice. Cycles are not a problem for the justifiability model. They are, in fact, quite informative: they are used to show that all justifiable preferences agree on the rankings of particular pairs of items. Ruling out cycles would significantly reduce the scope of the model. Furthermore, Corollary \ref{cor:maximal} does not require the introduction of outside knowledge about the DM. It is useful even if one does not believe in, or wish to introduce, non-choice data. 

Finally, the justifications in this model are easier to understand and interpret than the rationales in CFS. If a set of rationales is inconsistent with any set of transitive, monotone and complete preferences on the domain, can it really rationalize anything? To illustrate, consider the anecdote discussed at the end of CFS. A man goes into a store to buy an encyclopedia. He is presented with two desirable encyclopedias, but leaves without buying either. He feels certain that he would have purchased either encyclopedia if it had been presented in isolation. The justifiability model cannot accommodate this behavior. Formally, the domain is an almost-WARP set. The item selected from this set (don't buy) is pairwise-beaten by both of the other alternatives. IEA says this situation cannot happen. It is consistent with CFS, though---the DM can have a rationale for buying encyclopedia 1 over not buying, and a rationale for buying encyclopedia 2 over not buying, but not a rationale for buying encyclopedia 1 over encyclopedia 2. But there is something strange about this explanation. It requires the DM to have (and use) a rationale for not buying over buying either encyclopedia. Since we are told that both encyclopedias are high quality, and that the DM originally intended to buy one, it is hard to see how this could be a genuine ``rationale,'' or why the DM would appeal to it. Perhaps there is something missing from the exposition of the problem, e.g. fear of making a poor choice and regretting it. Since the rationalizability model does not shed much light on situations like these, not much is lost by excluding them. Any loss is offset by improved tractability, interpretation and identification. 

\section{Further Work}
The model presented in this paper has a lexicographic form. The DM's first priority is to avoid doing something unjustifiable, and his second priority is to do what he wants. Like all models of this form, the model is vulnerable to the criticism that it is too extreme. Perhaps the DM should weigh the cost of pretending to be a better person against the cost of deviating from his preferred choice.

To illustrate, consider a stylized version of the results from \citet{berger1997effect}. The DM chooses whether to make a donation to a university. If he is offered the choice between a small donation $s$ and no donation $n$, he chooses the small donation. If a large donation $\ell$ is offered to the choice set, he chooses no donation. The justifiability model can generate these choices if the DM has $s \succ n \succ \ell$ but thinks that the only justifiable preferences are $\ell \succ_1 s \succ_1 n$ and $n \succ_2 s \succ_2 \ell$. The DM might think that $\succ_2$ is justifiable, but $\succ$ is not, because someone with ``good'' reasons for preferring $n$ to $\ell$ (e.g. concerns about how the money will be spent, or a commitment to more important causes) would also prefer $n$ to $s$. 

This explanation works, but it is probably not the most natural one. The alternative explanation runs as follows: the DM doesn't want to donate at all, but he likes looking like a generous person. But if pooling with generous people gets too expensive, he decides not to donate and accepts looking like a cheapskate. The justifiability model cannot accommodate this argument. Either the DM thinks that preferring $n$ to $s$ is unjustifiable (in which case he always makes a donation) or he doesn't (in which case he always declines). It would be interesting to extend the justifiability model to account for tradeoffs like these. It is unlikely that a result with the simplicity and generality of Theorem \ref{thm:main} is available for this more complicated model, though. \citet{cherepanov2013revealed} study a model like this, but they assume a constant benefit from acting like a good person, and that the actions a good person would choose are known to the analyst. They relax the latter assumption but do not provide a full characterization. \citet{dillenberger2012ashamed} is another paper in this direction. It uses an EU setting and a single social norm with a special structure. 

We conclude with a note on welfare implications. Although we have referred to the tiebreaker in the justifiability model as the ``true preference,'' we do not take it to be a measure of the DM's welfare. One can always draw a distinction between choice behavior and welfare---even when the DM's choices maximize a standard preference relation---but the point seems particularly important when moral principles are involved. It is entirely plausible that the DM is better off when he chooses in accordance with his principles rather than his baser instincts. Even if he lacks lofty moral sentiments, he may get a great deal of utility from having a virtuous social image. A well-intentioned policymaker who forces the DM to accept the $\succ$-best item robs him of these benefits, potentially making him worse off. $\succ$ should be seen as an interesting feature of the DM's psychology, not the preference of a benevolent agent acting on the DM's behalf. 

\bibliography{just}

\appendix

\section{Proofs of Results in Text}
\subsection{Proof of Theorem \ref{thm:main}}

First, we show necessity. Necessity of Optimization follows because the items that maximize a preference over a set must all be indifferent. For IUA, fix $A, a$ such that $a \in A$. Suppose $a \succsim c(A)$ and $a \notin c(A)$. Then, there must not be any $\succ_m \in \mathcal{M}$ that has $a \succ_m A \setminus \{a\}$. Said another way, for any $\succ_m \in \mathcal{M}$, there exists $x \in A \setminus \{a\}$ such that $x \succ_m a$. Take any $B \supset A$. Suppose $x \notin M(B \setminus \{a\})$ but $x \in M(B)$. When $a$ is added to $B \setminus \{a\}$, no unjustified item in $B \setminus \{a\}$ can become justified. (If it wasn't top-ranked in the smaller menu, it can't be top-ranked in the larger menu.) The only remaining possibility is $x = a$. Since $x$ is not top-ranked when $A$ is present, this possibility is ruled out too. Now suppose $x \in M(B)$ but $x \notin M(B \setminus \{a\})$. Then there must be some $\succ_m$ such that $a \succ_m x \succ_m B \setminus \{a, x\}$. But since $a$ is always ranked below something in $A$, there is no such $\succ_m$.

Now we show sufficiency. Optimization says that everything in $c(X)$ must be indifferent. Thus, there are three types of items we need to consider: (1) items at least as good as everything in $c(X)$ but not in it, (2) items in $c(X)$, and (3) items worse than everything in $c(X)$. We need to construct $\mathcal{M}$ so that items in group (1) are never in $M(X)$ (first step), and items in group (2) are always in $M(X)$ (second step). It doesn't matter whether items in group (3) are in $M(X)$ or not. 

Let $\Pi$ denote the set of total orders over $\mathcal{A}$. Let 
\[\mathcal{M} := \{\succ_m \in \Pi: \forall X, y \text{ s.t. } X \text{ excludes } y, \exists x \in X \text{ s.t. } x \succ_m y\}.\]
We have chosen $\mathcal{M}$ so that
\[M(X) := \{x \in X: \exists \succ_m \in \mathcal{M} \text{ s.t. } x \succ_m X \setminus \{x\}\}\]
cannot include any item $y$ such that $\{y\} \succsim c(X)$ but $y \notin c(X)$. (For any such $y$, $X \setminus \{y\} \text{ excludes } y$, so $y$ cannot be ranked above $X \setminus \{y\}$ by any $\succ_m \in \mathcal{M}$.) 

We still need to show that $M(X)$ includes each item in $c(X)$. It is helpful to express the definition of $\mathcal{M}$ in a different way.

\begin{definition}[Exclusion from below]
$A$ excludes $b$ from below (written $A \triangleright b$) if $b \notin c(A \cup \{b\})$ and $b \succsim A$. 
\end{definition}

\begin{definition}[Menu-item relation]
A menu-item relation is a subset of $\mathcal{F}(\mathcal{A}) \cup \{\emptyset\} \times \mathcal{A}$.
\end{definition}

Clearly, $\triangleright$ is a menu-item relation, but it is not the only one that will appear in this proof. It is helpful to extend the definition of transitivity to menu-item relations. 

\begin{definition}[Transitivity]
A menu-item relation $R$ is transitive if 
\[\left(X \; R \; x, Y \; R \; y \text{ and } x \in Y\right) \Longrightarrow (X \cup Y) \setminus \{x, y\} \; R \; y.\]
\end{definition}

$\triangleright$ is transitive. To see why, suppose $X \triangleright x$ and $Y \triangleright y$ and $x \in Y$. By definition of $\triangleright$, $x \succsim x'$ for all $x' \in X$, and $y \succsim y'$ for all $y' \in Y$. Thus, $y \succsim z$ for all $z \in X \cup Y \setminus \{x, y\}$. It remains to show that $y \notin c(X \cup Y \setminus \{x\})$. Suppose $y \in c(X \cup Y \setminus \{x\})$. By IUA and $X \triangleright x$, $y \in c(X \cup Y)$. By IUA and $Y \triangleright y$, $y \in c(X \cup Y \setminus \{y\})$, a contradiction. 

Lemma \ref{lem:exc_from_below} states another useful property of $\triangleright$. The full strength of Lemma \ref{lem:exc_from_below} is actually not needed here, but will be useful later. 

\begin{lemma}
\label{lem:exc_from_below}
Fix $A \in \mathcal{F}(\mathcal{A})$ and $a \in A$ such that $A \setminus \{a\}$ excludes $a$. Let
\[\underline{A} := c(A) \cup \{a \in A: c(A) \succ a\}.\]
Then, $\underline{A} \triangleright a$. 
\end{lemma}

\begin{proof}
First, $A \setminus \{a\}$ excludes $a$ implies $a \succsim c(A)$ and $a \notin c(A)$. By definition of $\underline{A}$, we have $a \succsim \underline{A}$ and $a \notin \underline{A}$. We want to show that $a \notin c(\underline{A} \cup \{a\})$. 

Let $\bar{A}$ be the set of items that are weakly better than $c(A)$, but not in $c(A)$. Clearly, $\bar{A} = A \setminus \underline{A}$. Enumerate the items in $A$ from best (1) to worst ($|A|$). Ties can be broken arbitrarily, with one exception: everything in $\bar{A}$ must come before everything in $\underline{A}$. 

We show that each $\bar{a} \in \bar{A}$ is excluded from below by the items that succeed it. Suppose $a_1 \in \bar{A}$. Clearly, $A \setminus \{a_1\} \triangleright a_1^1$. Now suppose $a_2 \in \bar{A}$. By IUA, we can remove $a_1$ from $A$ without changing choice, so $a_2$ is still not chosen. We have $A \setminus \{a_1, a_2\} \triangleright a_2$. Iterating this argument, we get the desired result for every item in $\bar{A}$.

Now suppose that $a_i$ is the last item in $\bar{A}$. We have just shown that $\underline{A} \triangleright a_i$. Now consider $a_{i-1}$. We have also shown that $\{a_i\} \cup \underline{A} \triangleright a_{i-1}$. Transitivity of $\triangleright$ implies $\underline{A} \triangleright a_{i-1}$. We can apply the same argument to $a_{i-2}$, and so on. Once we get to $\underline{A} \triangleright a_1$, we are done. 
\end{proof}

We are finally ready to redefine $\mathcal{M}$:
\begin{equation}
\mathcal{M} = \{\succ_m \in \Pi: \nexists A, b \text{ s.t. } b \succ_m A \text{ and } A \triangleright b\}.
\label{eq:redefine_M}
\end{equation}
Suppose $A$ excludes $b$. By Lemma \ref{lem:exc_from_below}, $\underline{A} \triangleright b$. If we exclude any preference from $\mathcal{M}$ that has $b$ ranked above $\underline{A}$, then we automatically exclude any preference that has $b$ ranked above $A$. The set of total orders consistent with exclusion from below is no larger than the set of total orders consistent with exclusion.

Now we show that, for any $b \in \mathcal{A}$ and any $A \in \mathcal{F}(\mathcal{A})$ such that $b \notin A$:
\begin{equation}
\label{eq:final}
b \in c(A \cup \{b\}) \Longrightarrow \left(\exists \succ_m \in \mathcal{M} \text{ s.t. } b \succ_m A \right).
\end{equation}
We will construct an appropriate $\succ_m$. Let 
\begin{align}
\label{eq:LH}
L &:= A \cup \{x \in \mathcal{A}: \exists A' \subseteq A \text{ s.t. } A' \triangleright x\} \\
H &:= \mathcal{A} \setminus L.
\end{align}

For any $l \in L$ and $h \in H$, we will impose $h \succ_m l$. We will have $b \succ_m A$ provided $b \in H$. Suppose $b \in L$. Since $b \notin A$, there must exist $A' \subset A$ such that $A' \triangleright b$. By IUA, $b \notin c(A \cup \{b\})$, a contradiction.

Suppose we have $L' \subset L, h \in H$ such that $L' \triangleright h$. We can write $L' = A' \cup B'$, where $A' \subseteq A$ and each element of $B'$ is excluded from below by a subset of $A$. By transitivity of $\triangleright$, $h$ is excluded from below by a subset of $A$. This contradicts $h \in H$. Thus, there is no $L' \subset L$, $h \in H$ such that $L' \triangleright h$.

We need to find a way of ordering the items \textit{within} $H$ and $L$. Focus on $H$, since the argument for $L$ is the same. We want to find a preference $\succ_H$ over $H$ that is consistent with $\triangleright$. For any $X \in \mathcal{F}(H)$ and $y \in H$, we must have
\begin{equation}
\label{eq:H_order}
X \triangleright y \Longrightarrow \left(\exists x \in X \text{ s.t. } x \succ_H y \right).
\end{equation}

This step requires some groundwork. First, we introduce two new properties of menu-item relations like $\triangleright$. Clearly, $\triangleright$ will satisfy them. 

\begin{definition}[Properness]
A menu-item relation $R$ is proper if $X \; R \; x \Longrightarrow X \neq \emptyset$.
\end{definition}

\begin{definition}[Irreflexivity]
A menu-item relation $R$ is irreflexive if $X \; R \; x \Longrightarrow x \notin X$. 
\end{definition}

\begin{lemma}
\label{lem:closure}
Fix an irreflexive menu-item relation $R$. Let $R^0 := R$. For $i > 0$, let $R^i$ be the extension of $R^{i-1}$ obtained by imposing 
\[\left(\bigcup_{j=1}^k X_j \cup \{y_{k+1}, \ldots, y_n\} \right) \setminus \{y\} \; R^i \; y\]
whenever
\[\{y_1, \ldots, y_n\} \; R^0 \; y \text{ and, for all } j \leq k \leq n, \; X_j \; R^{i-1} \; y_j.\]
Then, the transitive closure of $R$ is $\bigcup_{i=0}^\infty R^i$. 
\end{lemma}

\begin{proof}
This is a standard result about the transitive closure. The usual proof goes through with the definition of transitivity used here. 
\end{proof}

\begin{lemma}
\label{lem:no_cycle}
Fix an irreflexive, transitive and proper menu-item relation $R$. Fix distinct $x, y \in \mathcal{A}$ such that $\neg (\{y\} \; R \; x)$. The transitive closure of $R \; \cup (\{x\}, y)$ is irreflexive and proper. 
\end{lemma}

\begin{proof}
Let $R' := R \; \cup (\{x\}, y)$, and let $R^+$ denote the transitive closure of $R'$. To deal with repeated applications of transitivity, we need the notion of a tree. For brevity, we write $\{z^i\}$ instead of $\{z_0, \ldots, z^i\}$. 

\begin{definition}[Q-tree]
\label{def:tree}
For a menu-item relation $Q$, a $Q$-tree starting at $w$ and ending at $W$ consists of:
\begin{itemize}
\item A bottom node $z_0 := w$, which is mapped to a parent set $Z_1(z_0)$ such that $Z_1 \; Q \; z_0$.
\item For $i > 0$: each $z_i(z^{i-1}) \in Z_i(z^{i-1})$ such that $z_i(z^{i-1}) \notin W \cup \{z^{i-1}\}$ is mapped to a parent set $Z_{i+1}(z^i)$ such that $Z^{i+1}(z^i) \; Q \; z_i$.
\item For some finite $K > 0$: each $z_K(z^{K-1}) \in W \cup \{z^{K-1}\}$. 
\end{itemize}
We refer to the $z_i$ as the $i$-th level of the tree, and we refer to nodes that do not have any parents as top nodes. All the nodes in level $K$ are top nodes, but top nodes may also appear in lower levels (except level $0$). A branch of the tree is a sequence $(z_0, \ldots, z_k)$ in which $z_i$ is the parent of $z_{i-1}$, $z_k$ is a top node, and $z_0$ is the bottom node. We refer to $(z_0, \ldots, z_{i-1})$ as the descendants of $z_i$, and $(z_{i+1}, \ldots, z_k)$ as the ancestors of $z_i$. 
\end{definition}

By Lemma \ref{lem:closure}, any pair $(W, w)$ in $R^+$ can be obtained by repeated application of transitivity. That is, there must be a $R'$-tree starting at $w$ and ending at $W$. It is obvious that $R^+$ will be irreflexive, so we focus on properness. Suppose there is a $R'$-tree starting at $w$ and ending at $\emptyset$. Every branch $(z_0, \ldots, z_k)$ of this tree must have $z_k = z_j$ for some $0 < j < k-2$. We will derive a contradiction by working through this tree. 

First, notice that the tree must use $\{x\} \; R' \; y$. If it doesn't, then it is an $R$-tree. We can work backward through the tree to derive a contradiction. Start by looking at levels $K-2$ through $K$. Take any $z_{K-1}(z^{K-2})$ that is not already a top node. By transitivity, we can replace it with $Z_K(z^{K-1}) \setminus \{z_{K-2}\}$. Now the tree has $K-1$ levels. Iterating this process, we should end up with a 1-level tree, where the set of level-1 nodes is a subset of the top nodes from the original tree. But each top node in the original tree was identical to some node further down the branch, so it must have been eliminated. Thus, level 1 is empty, and we must have $\emptyset \; R \; w$. Since $R$ is proper, this is a contradiction. 

Now prune the tree as follows. Wherever $x$ is the sole parent of $y$, remove $x$ and all its ancestors. The pruned tree does not use $\{x\} \; R' \; y$, so all the parental relationships in the pruned tree are from $R$. Take any pruned branch. We know that $y$ occurs at the top of this branch and nowhere further down. (If it did, then the branch would have terminated with $y$, and would not have needed to be pruned.) By working through the pruned tree as above, we get a 1-level tree in which $w$ is the bottom node. We know $y$ is one of the top nodes. If there are any other top nodes, they must be top nodes of the original, unpruned tree. But any such node would get deleted as we worked backwards, so $y$ must be the unique top node. We have $\{y\} \; R \; w$. 

Now take any level-1 node $v$ of the pruned tree that occurs below $y$ along a pruned branch. Consider the subtree of the pruned tree in which $v$ is the bottom node. By working through this subtree as above, we get a 1-level $R$-tree in which $v$ is the bottom node. We know $y$ is one of the top nodes. $w$ may also be a top node. ($w$ occurs below $v$, so it would not get deleted as we worked backwards.) But we have already seen that $\{y\} \; R \; w$, so $\{y\} \; R \; v$ by transitivity. We can make similar arguments for any level-2 node of the pruned tree that occurs below $y$ along a pruned branch. We continue in this way until we have shown that $\{y\} \; R \; v$ for any node $v$ that occurs below $y$ on a pruned branch. 

Now consider the subtrees we pruned from the original tree. Each has bottom node $x$. None of them can use $\{x\} \; R' \; y$. (Otherwise, some branch would contain duplicates of both $x$ and $y$, which is impossible.) All the parental relationships in the subtree are from $R$. We can work backward through the subtree in the familiar way, creating a new 1-level $R$-tree that has bottom node $x$. The set of top nodes must be non-empty because $R$ is proper. It may contain $y$ and any other node $v$ of the original tree that recurs below $y$ in that tree. But we have already seen that $\{y\} \; R \; v$ for any such $v$, so transitivity delivers $\{y\} \; R \; x$. This contradicts our assumption about $x, y$. 
\end{proof}

Now we use Lemma \ref{lem:no_cycle} to show that $\triangleright$ can be extended to an irreflexive, proper and transitive relation $\triangleright^+$ such that, for all $x, y \in \mathcal{A}$, $\{x\} \triangleright^+ y$ or $\{y\} \triangleright^+ x$. The proof is similar to that of the Szpilrajn Extension Theorem. Consider the set of irreflexive, proper and transitive relations that extend $\triangleright$, ordered by set inclusion. Take any chain in the partially ordered set. The union of its elements is clearly irreflexive, proper and transitive, so it is an upper bound for the chain. By Zorn's Lemma, the partially ordered set must have a maximal element $\triangleright^+$. Suppose that, for some $x, y$, neither $\{x\} \triangleright^+ y$ nor $\{x\} \triangleright^+ y$. By Lemma \ref{lem:no_cycle}, $\triangleright^+$ can be extended to another irreflexive, proper and transitive relation containing $(\{x\}, y)$. Then $\triangleright^+$ cannot be maximal, a contradiction. Moreover, for each $X \in \mathcal{F}(\mathcal{A})$ and $y \in \mathcal{A}$, $\triangleright^+$ must satisfy
\begin{equation}
\label{eq:plus}
X \triangleright y \Longrightarrow \left(\exists x \in X \text{ s.t. } \{x\} \triangleright^+ y\right).
\end{equation}
Suppose not. Then $\{y\} \triangleright^+ x$ for all $x \in X$, as well as $X \triangleright^+ y$. Since $\triangleright^+$ is transitive, $\emptyset \triangleright^+ y$. Since $\triangleright^+$ is proper, this is a contradiction. 
Similarly, suppose that $\{x\} \triangleright^+ y$ and $\{y\} \triangleright^+ x$. By transitivity, $\emptyset \; \triangleright^+ x$, a contradiction. 

Now we are ready to order the items in $H$. We can use $\triangleright^+$ to define $\succ_H$:
\[\forall x, y \in H \quad \{x\} \triangleright^+ y \Longrightarrow x \succ_H y.\]
$\succ_H$ is antisymmetric, complete and transitive. By (\ref{eq:plus}), it satisfies (\ref{eq:H_order}). 

We can use the same arguments to obtain an appropriate $\succ_L$. Now that we have a total order over $H$ and total order over $L$, and we know that everything in $H$ must be strictly better than everything in $L$, we have a total order $\succ_m$ over the whole of $\mathcal{A}$. We have seen that, for any $X \subset \mathcal{F}(\mathcal{A})$ and any $y \in \mathcal{Y}$, $\succ_m$ satisfies
\begin{equation}
\label{eq:total_order}
X \triangleright y \Longrightarrow \left(\exists x \in X \text{ s.t. } x \succ_m y \right).
\end{equation}
Thus, $\succ_m$ is indeed in $\mathcal{M}$. Since $b \succ_m A$, (\ref{eq:final}) holds, and we are done. 

\subsection{Proof of Proposition \ref{prop:dominance}}

First, we show necessity of ISA. Suppose $(\succsim, c)$ has a monotone justifiability representation. It suffices to show that $M(B) = M(B \setminus A)$. Suppose $x \in M(B \setminus A)$ but $x \notin M(B)$. Then there must be $\succ_m \in \mathcal{M}$ that ranks some $a \in A$ above everything in $B \setminus A$. Fix some such $\succ_m$, and let $a^*$ be the $\succ_m$-best element of $A$. We have $a^* \succ_m B \setminus \{a^*\}$. But since $a^*$ is excluded by a subset of $B$ or dominated by an element of $B$, this is impossible. Now suppose $x \in M(B)$ but $x \notin M(B \setminus A)$. This is clearly impossible unless $x = a$ for some $a \in A$. But we have just seen that no $a \in A$ can be top-ranked in $B$, so $x = a$ is impossible too. 

Now we show sufficiency. It will be convenient to have a weak dominance relation $\succsim_D$ as well as a strict dominance relation $\succ_D$. Say that $a \succsim_D b$ if $a = b$ or $a \succ_D b$. Thus, $\succsim_D$ is reflexive. 

As before, let $\Pi$ denote the set of total orders on $\mathcal{A}$. Let 
\begin{align*}
\Pi^D &:= \{\succ_m \in \Pi: \nexists a, b \text{ s.t. } a \succ_m b \text{ and } b \succ_D a\} \\
\mathcal{M} &:= \{\succ_m \in \Pi^D: \nexists A, b \text{ s.t. } b \succ_m A \text{ and } A \text{ excludes } b\}.
\end{align*}

We know that $M(X)$ will not include any item $y$ such that $y \succsim c(X)$ but $y \notin c(X)$. We still have to show that $M(X)$ includes each item in $c(X)$. It is helpful to express the definition of $\mathcal{M}$ in a different way. 
\begin{definition}[$D$-exclusion]
$A \in \mathcal{F}(\mathcal{A})$ $D$-excludes $b \in A$ (written $A \triangleright_D b$) if $A$ excludes $b$ and, for all $a \in A$, $\neg(b \succ_D a)$. 
\end{definition}

Suppose that $A$ excludes $b$. Let
\[D(b, A) := \{a \in A: \neg(b \succ_D a)\}.\]
We claim that $D(b; A) \triangleright_D b$. We need $b \succsim c(D(b, A) \cup \{b\})$ and $b \notin c(D(b, A) \cup \{b\})$. By ISA, $c(D(b, A) \cup \{b\}) = c(A \cup \{b\})$. Since $b \succsim c(A \cup \{b\})$ and $b \notin c(A \cup \{b\})$, the result follows.

This allows us to rewrite $\mathcal{M}$:
\[\mathcal{M} = \{\succ_m \in \Pi^D: \nexists A, b \text{ s.t. } b \succ_m A \text{ and } A \; \triangleright_D b\}.\]

Now we show that $b \in c(A \cup \{b\})$ implies $b \succ_m A$ for some $\succ_m \in \mathcal{M}$. We will construct an appropriate $\succ_m$.

This step requires some extra work because $\triangleright_D \cup D$ does not satisfy any nice transitivity property, so we have to take a transitive closure.\footnote{I am abusing notation here. Strictly speaking, $D$ is not a menu-item relation, but a relation on $\mathcal{A}$. But of course there is a menu-item relation that is isomorphic to $D$.} Then we make sure the transitive closure satisfies other nice properties. 

\begin{definition}[$D$-transitivity]
A menu-item relation $R$ is $D$-transitive if 
\[(X \; R \; x, Y \; R \; y \text{ and } x \in Y) \Longrightarrow D(y, X \cup Y \setminus \{x, y\}) \; R \; y.\]
\end{definition}

\begin{definition}[$D$-monotonicity]
A menu-item relation $R$ is $D$-monotone if it satisfies
\begin{align*}
&x \; \succ_D \ ; y \Longrightarrow x \; R \; y \\
\text{and} \quad &X \; R \; y \Longrightarrow (\nexists x \in X \text{ s.t. } y \succsim_D x).
\end{align*}
\end{definition}
Since $\succsim_D$ is reflexive, the second part of $D$-monotonicity implies Irreflexivity. Notice that $\triangleright_D \cup \succ_D$ is $D$-monotone.

\begin{lemma}
\label{lem:D-closure}
The $D$-transitive closure of $\triangleright_D \cup \succ_D$ is $D$-monotone and proper. 
\end{lemma}

\begin{proof}
We can show a result analogous to Lemma \ref{lem:closure} for $D$-transitivity, so anything in the $D$-transitive closure is obtained via repeated application of $D$-transitivity. Clearly, application of $D$-transitivity will not lead to a violation of $D$-monotonicity. As before, we check properness via a tree. The relevant construction is very similar to Definition \ref{def:tree}. The only difference is as follows. In Definition \ref{def:tree}, a node becomes a top node if and only if it is in $W$ or is identical to one of its descendants. Now, a node becomes a top node if and only if it satisfies one of those conditions \textit{or} it is dominated by one of its descendants. 

Suppose the $D$-transitive closure of $\triangleright_D \cup \succ_D$ is improper, so we have a tree starting at $w$ and ending at $\emptyset$. Consider a menu that consists of all the nodes in the tree. By assumption, each item in the menu is excluded by its parents or dominated by its parent. By ISA, all these items can be removed without changing choice. Since some item must be chosen, this is a contradiction. 
\end{proof}

Denote the $D$-transitive closure of $\triangleright_D \cup \succ_D$ by $R_0$. We can now use $R_0$ to define $H^D$ and $L^D$, the analogues of $H$ and $L$ in (\ref{eq:LH}). Let
\begin{align*}
L^D &:= A \cup \{x \in \mathcal{A}: \exists A' \subseteq A \text{ s.t. } A' \; R_0 \; x\} \\
H^D &:= \mathcal{A} \setminus L
\end{align*}
For $l \in L^D$ and $h \in H^D$, we will impose $h \succ_m l$. We will have $b \succ_m A$ provided $b \in H^D$. Suppose $b \in L$, so $A' \; R_0 \; b$ for some $A' \subseteq A$. Then we have a tree starting at $b$ and ending at $A'$. Consider a menu that consists of all the nodes in the tree as well as $A \setminus A'$. Each item in the menu that is not identical to some element of $A$ is sub-maximal. In particular, $b$ is sub-maximal, so it cannot be chosen from this menu. Now remove all the sub-maximal elements of the menu except $b$ and elements of $A$. By ISA, $b$ is still not chosen: $b \notin c(A \cup \{b\})$. This is a contradiction. 
Since $R_0$ is $D$-transitive, there cannot be $L' \subset L^D, h \in H^D$ such that $L' \triangleright_D h$. Similarly, there cannot be $l \in L^D, h \in H^D$ such that $l \succ_D h$. 

Now we need to find a way of ordering the items within $H$ and $L$. We want to find a preference $\succ_H$ over $H$ that is consistent with $\triangleright_D$ and $\succ_D$. For any $X \in \mathcal{F}(H)$ and $y \in H$, we must have
\[X \triangleright_D y \Longrightarrow (\exists x \in X \text{ s.t. } x \succ_H y).\]
Also, for any $x, y \in H$, we must have
\[x \succ_D y \Longrightarrow x \succ_H y.\]

This part of the argument is almost the same as Theorem \ref{thm:main}. The idea is to extend $R_0$ in the same way that we extended $\triangleright$. $R_0$ is already consistent with $\succ_D$, so we just have to worry about $\triangleright_D$. 

\begin{lemma}
\label{lem:D-extension}
Fix a $D$-monotone, $D$-transitive and proper menu-item relaton $R$. Fix distinct $x, y \in \mathcal{A}$ such that $\neg(\{y\} \; R \; x)$. The $D$-transitive closure of $R \cup (\{x\}, y)$ is $D$-monotone and proper.
\end{lemma}

\begin{proof}
Let $R' := R \cup (\{x\}, y)$, and let $R^+$ denote the $D$-transitive closure of $R'$. 

The proof is very similar to that of Lemma \ref{lem:no_cycle}. Any pair $(W, w)$ in $R^+$ can be obtained by repeated application of $D$-transitivity. Clearly, application of $D$-transitivity will not lead to a violation of $D$-monotonicity. For properness, we need to make sure there is no tree starting at $w$ and ending at $\emptyset$. Recall from the proof of Lemma \ref{lem:D-closure} that a top node in such a tree must be identical to one of its descendants or dominated by one of its descendants. 

Suppose there is a tree starting at $w$ and ending at $\emptyset$. As before, we can show the tree must use $\{x\} \; R' \; y$. Then we cut each branch where $x$ is the sole parent of $y$, and show that $\{y\} \; R \; v$ for any descendant $y$ along a pruned branch. Next, we work backward through the subtrees cut from the original tree. We produce 1-level $R$-trees in which $x$ is the bottom node and the top nodes may include $y$, descendants of $y$ along a pruned branch, and items dominated by $y$ or the aforementioned descendants. Since $R$ is $D$-monotone, $\{y\} \; R \; v$ for any item $v$ that $y$ dominates. We have already seen that $\{y\} \; R \; v$ for any item $v$ that descends from $y$ along a pruned branch. Finally, since $R$ is $D$-monotone and $D$-transitive, $v \succ_D u$ and $\{y\} \; R \; v$ implies $\{y\} \; R \; u$. Putting all this together and applying $D$-transitivity, we get $\{y\} \; R \; x$, a contradiction. 
\end{proof}

In Theorem \ref{thm:main}, we saw that $\triangleright$ could be extended to a transitive, irreflexive and proper relation $\triangleright^+$ such that
\begin{align*}
\forall x, y \in \mathcal{A} \quad &\{x\} \; \triangleright^+ y \text{ or } \{y\} \; \triangleright^+ x \\
\text{and } \forall X \in \mathcal{F}(\mathcal{A}), y \in \mathcal{A} \quad &X \triangleright y \Longrightarrow (\exists x \in X \text{ s.t. } \{x\} \triangleright^+ y).
\end{align*}
We can use exactly the same arguments to show that $R_0$ can be extended to a $D$-monotone, $D$-transitive and proper relation $R^+$ that satisfies the same properties (with $R^+$ in place of $\triangleright^+$ and $R_0$ in place of $\triangleright$). We can use this to define $\succ_H$:
\[\forall x, y \in H \quad \{x\} \; R^+ \; y \Longrightarrow x \succ_H y.\]
$\succ_H$ is complete and antisymmetric. For transitivity, suppose $\{a\} \; R^+ \; b$ and $\{b\} \; R^+ \; c$. Since $R^+$ is $D$-transitive, we will have $\{a\} \; R^+ \; c$ provided $\neg(c \succ_D a)$. Suppose $c \succ_D a$. Since $R^+$ is $D$-monotone, $\{c\} \; R^+ \; a$. Applying $D$-transitivity, we have $\{c\} \; R^+ \; b$. Applying $D$-transitivity again, we have $\emptyset \; R^+ \; c$. This contradicts the properness of $R^+$. 

We can use the same arguments to obtain an appropriate $\succ_L$. As in Theorem \ref{thm:main}, we use $\succ_L$ and $\succ_H$ to obtain the desired total order $\succ_m$ over $\mathcal{A}$. 

\subsection{Proof of Theorem \ref{thm:EU}}
First, we show necessity. It is well known that the first parts of Continuity and Independence are necessary (and sufficient) for $\succsim$ to have an EU representation. Consider the second part of Continuity. Take any convergent sequence $(q_n)$ such that each $q_n \in NB(p)$. By definition, $q_n \succsim p$ and $p \in c(\{q_n, p\})$ for all $n$. Since $\succsim$ is continuous, $q \succsim p$. For each $n$, we have $\mathbb{E}_p m_n \geq \mathbb{E}_{q_n} m_n$ for some $m_n \in \mathcal{M}^{EU}$. Since $\mathcal{M}^{EU}$ is compact, some subsequence of $m_n$ has a limit $m \in \mathcal{M}^{EU}$. We will have $\mathbb{E}_p m \geq \mathbb{E}_q m$, so $p \in c(\{q, p\})$. We conclude that $q \in NB(p)$, so $NB(p)$ is closed.

Now consider Convexity. Suppose that $A$ excludes $p$. By Lemma \ref{lem:exc_from_below}, $p \succsim \underline{A}$ and $p \notin c(\underline{A} \cup \{p\})$. For each $m \in \mathcal{M}^{EU}$, we have some $a \in \underline{A}$ such that $\mathbb{E}_p m < \mathbb{E}_a$. It is without loss to assume that $\mathbb{E}_p m = 0$ for all $m \in \mathcal{M}^{EU}$. We want to find a set of weights $\alpha$ such that $p \succsim \sum_{a \in \underline{A}} \alpha(a)\delta_a$ and 
\[\sum_{a \in \underline{A}}\alpha(a)\mathbb{E}_a m > 0\]
for all $m \in \mathcal{M}^{EU}$. The first part is easy---it will hold for any $\alpha$ since $\succsim$ is EU---so we focus on the second. For each $m \in \mathcal{M}^{EU}$, let $m_A := (\mathbb{E}_a m)_{a \in \underline{A}}$. Let $M_A$ be the set of the $m_A$. Like $\mathcal{M}^{EU}$, $M_A$ is nonempty, compact and convex. Let $N := \mathbb{R}^{|\underline{A}|}_-$, which is nonempty, closed and convex. Notice that no element of $M_A$ can be weakly negative (otherwise, some $m \in M$ would rank $p$ weakly higher than each member of $\underline{A}$). Thus, $M_A$ and $N$ are disjoint, and we can apply the Separating Hyperplane Theorem. This delivers a nonzero $\alpha \in \mathbb{R}^{|\underline{A}|}$ and $c \in \mathbb{R}$ such that $\alpha' n < c < \alpha' m_A$ for all $n \in N, m_A \in M_A$. Since the zero vector belongs to $N$, we must have $c > 0$. Suppose the $i$th element of $\alpha$ is strictly negative. By choosing $n$ with a sufficiently negative number in $i$th position and zeros elsewhere, we get $\alpha' n > c$, a contradiction. Thus, each element of $\alpha$ is weakly positive. If we rescale $\alpha$ to a unit sum, we still have $\alpha' m_A > 0$ for all $m_A \in M_A$. We can rewrite this as $\sum_{a \in \underline{A}} \alpha(a)\mathbb{E}_a m > 0$ for all $m \in M$, which is exactly what we needed. 

For the other direction of Convexity, suppose $p \in c(A \cup \{p\})$. For some $m \in \mathcal{M}^{EU}$, we have $\mathbb{E}_p m \geq \mathbb{E}_a m$ for all $a \in A$. For this $m$, we clearly have $\mathbb{E}_p m \geq \mathbb{E}_a m$ for all $a \in \text{co}(A)$. No $a \in \text{co}(A)$ can belong to $B(p)$. 

Consider the second part of Monotonicity. Suppose that $p >_{FOSD} q$. Take any menu $A \ni p, q$. Consider how $M(A)$ compares to $M(A \setminus \{q\})$. Since the justifiable preferences can be weakly FOSD-monotone, $q$ might be in $M(A)$. However, no item other than $q$ can be in $M(A)$ but not $M(A \setminus \{q\})$. Since the true preference is strictly FOSD-monotone, $q$ would not be chosen anyway, so it doesn't matter whether it is present. Now suppose some item $r$ is in $M(A \setminus \{q\})$ but not in $M(A)$. Then, some $m \in \mathcal{M}^{EU}$ must rank $q$ strictly above $r$, and $r$ weakly above everything else in $A$. Since $m$ must be weakly FOSD-monotone and $p \in A$, this is impossible. We must have $M(A) \setminus \{q\} = M(A \setminus \{q\})$. 

Now we show sufficiency. Recall the definition of $B(p)$ in (\ref{eq:Bp}). Similarly, let
\[W(p) := \{q \in \Delta(Z): q \succsim p \text{ and } q \notin c(\{p, q\})\}.\]
\begin{lemma}
$B(p)$ and $W(p)$ are convex cones.
\end{lemma}

\begin{proof}
First, suppose that $q \in B(p)$, so $p \succsim q$ and $\{q\} = c(\{p, q\})$. By $\succsim$-Independence, $p \succsim \alpha p + (1-\alpha) q$ for all $\alpha \in (0, 1)$. By $c$-Independence, $\{\alpha p + (1-\alpha)q\} = c(\{p, \alpha p + (1-\alpha)q\})$, so $\alpha p + (1-\alpha)q \in B(p)$. Similarly, suppose that $\alpha p + (1-\alpha)q \in B(p)$, so $p \succsim \alpha p + (1-\alpha)q$ and $\{\alpha p + (1-\alpha)q\} = c(\{\alpha p + (1-\alpha)q, p\})$. By $\succsim$-Independence, $p \succsim q$. By $c$-Independence, $\{q\} = c(\{q, p\})$, so $q \in B(p)$. Now suppose that $q, r \in B(p)$, so $p \succsim q, r$ and $\{q\} = c(\{p, q\})$, $\{r\} = c(\{p, r\})$. By $\succsim$-Independence, $p \succsim \alpha q + (1-\alpha)r$ for all $\alpha \in (0, 1)$. By $\succsim$- and $c$-Independence, 
\begin{align*}
\alpha p + (1-\alpha)r &\succsim \alpha q + (1-\alpha)r \\
\{\alpha q + (1-\alpha)r\} &= c(\{\alpha q + (1-\alpha)r, \alpha p + (1-\alpha)r\}),
\end{align*}
so $\alpha q + (1-\alpha)r$ excludes $\alpha p + (1-\alpha)r$ from below. Also by $\succsim$- and $c$-Independence, 
\begin{align*}
p &\succsim \alpha p + (1-\alpha)r \\
\{\alpha p + (1-\alpha)r\} &= c(\{p, \alpha p + (1-\alpha)r\}),
\end{align*}
so $\alpha p + (1-\alpha)r$ excludes $p$ from below. By IUA, we can add $p$ to a set containing $\alpha p + (1-\alpha)r$ without affecting choice, so 
\[\{\alpha q + (1-\alpha)r\} = c(\{\alpha p + (1-\alpha)r, \alpha q + (1-\alpha)r, p\}).\]
Also by IUA, we can remove $\alpha p + (1-\alpha)r$ from a set containing $\alpha q + (1-\alpha)r$ without affecting choice, so
\[\{\alpha q + (1-\alpha)r\} = c(\{\alpha q + (1-\alpha)r, p\})\]
so $\alpha q + (1-\alpha)r \in B(p)$. Parallel arguments suffice to show that $W(p)$ is a cone. 
\end{proof}

\begin{lemma}
\label{lem:shift}
Fix $p, p', q, q' \in \Delta(Z)$ such that $q - p = q' - p'$. If $q \in B(p)$, then $q' \in B(p')$. If $q \in W(p)$, then $q' \in W(p')$. 
\end{lemma}

\begin{proof}
We have
\[\frac{1}{2}p + \frac{1}{2}q' = \frac{1}{2}p' + \frac{1}{2}q.\]
Since $q \in B(p)$, $p \succsim q$ and $\{q\} = c(\{p, q\})$. By $c$-Independence, 
\begin{align*}
\left\{\frac{1}{2}q + \frac{1}{2}p'\right\} &= c\left(\left\{\frac{1}{2}q + \frac{1}{2}p', \frac{1}{2}p + \frac{1}{2}p'\right\}\right) \\
\left\{\frac{1}{2}p + \frac{1}{2}q'\right\} &= c\left(\left\{\frac{1}{2}p + \frac{1}{2}q', \frac{1}{2}p + \frac{1}{2}p'\right\}\right) \\
\{q'\} &= c(\{p', q'\}).
\end{align*}
Similarly, by $\succsim$-Independence, 
\begin{align*}
\frac{1}{2}p + \frac{1}{2}p' &\succsim \frac{1}{2}q + \frac{1}{2}p' \\
\frac{1}{2}p + \frac{1}{2}p' &\succsim \frac{1}{2}p + \frac{1}{2}q' \\
p' &\succsim q'.
\end{align*}
We conclude that $q' \in B(p')$. Parallel arguments suffice to show that $q \in W(p)$ and $p - q = p' - q'$ imply $q' \in W(p')$. 
\end{proof}

Now we establish the relationship between $B(p)$ and $W(p)$. Suppose that $p = \alpha q + (1-\alpha)r$ where $q \in B(p)$, so $\alpha q + (1-\alpha)r \succsim q$ and $\{q\} = c(\{\alpha q + (1-\alpha)r, q\})$. By $\succsim$-Independence, $r \succsim q$, so $r \succsim \alpha q + (1-\alpha)r$. By $c$-Independence, $\{q\} = c(\{q, r\})$, so $\{\alpha q + (1-\alpha)r\} = c(\{\alpha q + (1-\alpha)r, r\})$. We conclude that $r \in W(p)$. Parallel arguments suffice to show that $r \in W(p)$ implies $q \in B(p)$. 

We are now ready to define $\mathcal{M}^{EU}$. Take $p$ in the interior of $\Delta(Z)$. Take any supporting hyperplane $H$ of $B(p)$ that passes through some boundary point $b$ of $B(p)$ with $p \succ b$. Since $B(p)$ is a cone with vertex $p$, $H$ will also pass through $p$. Since $B(p)$ is open in $\{q \in \Delta(Z): p \succsim q\}$, $H$ cannot include any point in $B(p)$. Let $\mathcal{H}$ be the set with generic member $H$. For each $H$, take a unit-norm $m \in \mathbb{R}^{|Z|}$ such that $m'h = 0$ for all $h \in H$ (including $p$) and $m'q > 0$ for all $q \in B(p)$. Collect these $m$, and take the closed convex hull. This is $\mathcal{M}^{EU}$. 

We show that everything in $\mathcal{M}^{EU}$ is weakly $D$-monotone. It suffices to show that any preference with an indifference curve in $\mathcal{H}$ that has $B(p) \succ p$ is weakly $D$-monotone. (All the other preferences are combinations and/or limits of these, so will inherit weak $D$-monotonicity.) Suppose that some such preference has $q \succ q'$ even though $q' >_{FOSD} q$. By moving $q, q'$ closer together along the segment connecting them, and shifting both by a common factor, we can always get $q \in B(p)$ and $q' \in NB(p)$. Since $q' \in NB(p)$, $p \in c(\{p, q'\})$. Since $q \in B(p)$, $p \notin c(\{p, q, q'\})$. Since $q' >_{FOSD} q$, $c(\{p, q, q'\}) = c(\{p, q'\})$ by IDA. But the only possibility for $c(\{p, q, q'\})$ is $q'$, and $c(\{p, q'\}) = p$---contradiction. 

We now show
\begin{equation}
\label{eq:intersect}
B(p) = \bigcap_{m \in \mathcal{M}^{EU}}\{q \in \Delta(Z): m'q > m'p\} \cap \{q \in \Delta(Z): p \succsim q\}.
\end{equation}
This could be violated in two ways. First, there could be a boundary point $b$ of $B(p)$ that doesn't belong to $B(p)$, but isn't a limit of boundary points that are strictly worse than $p$. Clearly, $b \sim p$. Moreover, any sufficiently small perturbation $\tilde{b}$ of $b$ such that $\tilde{b} \prec p$ must belong to $B(p)$. Said another way, for any $q \prec p$, there must be some $\alpha$ sufficiently close to 1 such that
\[\alpha b + (1-\alpha) q \in B(p).\]
Take $r$ such that $p >_{FOSD} r$. ($p$ is interior, so some such $r$ must exist.) Since the true preference is strictly FOSD-monotone, we have $r \prec p$. We also have $\alpha \in (0, 1)$ such that
\[\alpha b + (1-\alpha) r \in B(p).\]
That is, $p \notin c(\{p, \alpha b + (1-\alpha)r\})$ even though $p \succsim \alpha b + (1-\alpha)r$. By Convexity, $p \notin c(\{p, b, r\})$. By Monotonicity and $p >_{FOSD} r$, $p \notin c(\{p, b\})$. This contradicts $b \notin B(p)$. 

Second, there could be a sequence $\{H\}_{n=1}^\infty$ of hyperplanes in $\mathcal{H}$ converging to $\{q \in \Delta(Z): p \sim q\}$ even though $B(p) \cap \{q \in \Delta(Z): p \sim q\}$ is nonempty. Recall that each hyperplane in $\mathcal{H}$ passes through some boundary point of $B(p)$ that is strictly worse than $p$. Take the sequence of such points corresponding to $\{H\}_{n=1}^\infty$. Passing to a subsequence if necessary, let $b$ be the limit of this sequence of points. For any sufficiently small perturbation $\tilde{b}$ of $b$ with $\tilde{b} \prec p$, we must have $\tilde{b} \in B(p)$. (Otherwise, $\{q \in \Delta(Z): p \sim q\}$ could not be the limit of $\{H\}_{n=1}^\infty$.) We can now apply the argument in the previous paragraph. Take any $r$ such that $p>_{FOSD} r$. We must have $\alpha b + (1-\alpha)r \in B(p)$ for some $\alpha \in (0, 1)$. Applying Convexity, $b \notin c(\{p, b, r\})$. Applying IDA, $b \notin c(\{p, b\})$, so $b \in B(p)$, so $b \notin NB(p)$. $b$ is a limit point of $NB(p)$, so this contradicts Continuity. 

Notice that it does not matter which $p$ we use to define $\mathcal{M}^{EU}$, since Lemma \ref{lem:shift} ensures that we will get the same set of utilities (up to an irrelevant additive constant) for any interior $p$. To finish the proof, we have to show that $\mathcal{M}^{EU}$ satisfies two conditions. First, no utility in $\mathcal{M}^{EU}$ would justify choosing an item that the DM doesn't choose, but likes as much as anything he does choose. Second, for any item the DM chooses, some utility in $\mathcal{M}^{EU}$ justifies it.

Consider the first part. Suppose $A$ excludes $q$. By Lemma \ref{lem:exc_from_below}, $q \notin c(\underline{A} \cup \{q\})$ and $q \succsim \underline{A}$. By Convexity, we can find $a^* \in \text{co}(\underline{A})$ such that $q \notin c(\{q, a^*\})$. Since $q \succsim a^*$, $a^* \in B(q)$. By (\ref{eq:intersect}), $m'a^* > m'q$ for all $m \in \mathcal{M}^{EU}$. For each $m \in \mathcal{M}^{EU}$, we must have $a \in \underline{A}$ such that $m'a > m'q$. This is exactly what we needed. 

For the second part, suppose $q \in c(A \cup \{q\})$. To start, suppose $q \succsim A$. Suppose that we cannot find $m \in \mathcal{M}^{EU}$ so that $m'q \geq m'a$ for all $a \in A$. Recall the argument we used to show necessity of Convexity. Since $\mathcal{M}^{EU}$ is compact and convex, we can use the same argument to find an $a^* \in \text{co}(A)$ such that
\[a^* \in \bigcap_{m \in \mathcal{M}^{EU}}\{r \in \Delta(Z): m'r > m'q\} \cap \{r \in \Delta(Z): q \succsim r.\}\]
By (\ref{eq:intersect}), $a^* \in \text{co}(A) \cap B(q)$. By Convexity, $q \notin c(A \cup \{q\})$, a contradiction. 

Now we relax the assumption that $q \succsim A$. We know $q \in c(\underline{A} \cup \{q\})$. (Suppose not. Then $q \notin c(A \cup \{q\})$ by IUA, a contradiction.) By the previous argument, we can find $m \in \mathcal{M}^{EU}$ such that $m'q \geq m'a$ for all $a \in \underline{A}$. Now take any item $\bar{a} \in \bar{A} := A \setminus \underline{A}$. By definition of $\bar{A}$, $\bar{a} \succsim \{q\} \cup \underline{A}$. By Lemma \ref{lem:exc_from_below}, $\bar{a} \notin c(\{\bar{a}, q\} \cup \underline{A})$. Thus, there is no $m \in \mathcal{M}^{EU}$ such that $m'\bar{a} \geq m'q$ and $m'\bar{a} \geq m'a$ for all $a \in \underline{A}$. By the same argument we used to show necessity of Convexity, we have $a^* \in \text{co}(\{q\} \cup \underline{A})$ such that $m'a^* > m'\bar{a}$ for all $m \in \mathcal{M}^{EU}$. By (\ref{eq:intersect}) and the relationship between $B$ and $W$, 
\[W(a^*) = \bigcap_{m \in \mathcal{M}^{EU}}\{r \in \Delta(Z): m'a^* > m'r\} \cap \{r \in \Delta(Z): r \succsim a^*\}.\]
Since $\bar{a} \succsim a^*$, we have $\bar{a} \in W(a^*)$. Since $m'q \geq m'a$ for all $a \in \underline{A} \cup \{q\}$, $m'q \geq m'a^*$. By definition of $W(a^*)$, $m'a^* > m'r$ for any $r \in W(a^*)$. Putting these two facts together, $m'q > m'r$ for all $r \in W(a^*)$. In particular, $m'q > m'\bar{a}$ as well as $m'q \geq m'a$ for all $a \in \underline{A}$. Since the same argument applies for all $\bar{a} \in \bar{A}$, we have $m'q \geq m'a$ for all $a \in A$. Thus, we have $m \in \mathcal{M}^{EU}$ that justifies the selection of $q$ from $A \cup \{q\}$. 

\subsection{Proof of Corollary \ref{cor:EU_unique}}
For the minimal set, recall the construction in the proof of Theorem \ref{thm:EU}. We start with $p$ in the interior of $\Delta(Z)$. Then we take the supporting hyperplanes of $B(p)$ that pass through some boundary point of $B(p)$ that is strictly worse than $p$. $\mathcal{H}$ is the set of such hyperplanes. Let $\bar{co}(\mathcal{H})$ be the closed convex hull of $\mathcal{H}$. The preferences in $\mathcal{M}^{EU}_\succ$ are precisely the EU preferences that have indifference curves in $\bar{co}(\mathcal{H})$ and that prefer $B(p)$ to $p$. 

Now we show that every set of justifiable preferences must contain all the preferences in $\mathcal{M}^{EU}_\succ$. Let $\mathcal{G}$ be a strict subset of $\mathcal{H}$, and take the closed convex hull $\bar{co}(\mathcal{G})$. If $\mathcal{H} \subset \bar{co}(\mathcal{G})$, we are back where we started---so assume that $H \in \mathcal{H}$ but $H \notin \bar{co}(\mathcal{G})$. Recall that $H$ passes through some boundary point $b$ of $B(p)$ that is strictly worse than $p$. Since $B(p)$ is open in $\{q \in \Delta(Z): p \succsim q\}$, $b \notin B(p)$. It is without loss to assume that $b$ does not belong to any member of $\bar{co}(\mathcal{G})$. (Suppose that each boundary point of $B(p)$ that is strictly worse than $p$ belongs to some member of $\bar{co}(\mathcal{G})$. Then, convexity of $\bar{co}(\mathcal{G})$ implies $\mathcal{H} \subset \bar{co}(\mathcal{G})$, a contradiction.) Take $\mathcal{N}^{EU}_\succ$ to be the set of EU preferences that have indifference curves in $\bar{co}(\mathcal{G})$ and prefer $B(p)$ to $p$. For all $\succsim_m \in \mathcal{N}^{EU}_\succ$, we have $b \succ_m p$. But since $p \in c(\{b, p\})$, any set of justifiable EU preferences must contain $\succsim_m$ such that $p \succsim_m b$. Thus, $\mathcal{N}^{EU}_\succ$ cannot be a set of justifiable EU preferences. The same argument applies to any compact, convex set of EU preferences that excludes something in $\mathcal{M}^{EU}_\succ$. 

Now we turn to the maximal set. We will modify the construction in the proof of Theorem \ref{thm:EU}. As before, take $p$ in the interior of $\Delta(Z)$. The preferences in the maximal set are the EU preferences that (1) have an indifference curve through $p$ that is a supporting hyperplane of $B(p)$, (2) prefer $B(p)$ to $p$, and (3) are weakly $D$-monotone. The set of all such preferences is closed and convex. Call it $\mathcal{M}^{EU}_{max}$. $\mathcal{M}^{EU}_{max}$ satisfies (\ref{eq:intersect}), so we can use the argument from the proof of Theorem \ref{thm:EU} to show that it works. It is maximal because any preference outside it must violate $D$-monotonicity or prefer $p$ to something in $B(p)$, which is impossible. 

\subsection{Proof of Corollary \ref{cor:comp_stat_EU}}
If $B_1(p) \supset B_2(p)$, the set of supporting hyperplanes of $B_1(p)$ is strictly smaller than the set of supporting hyperplanes of $B_2(p)$. Thus, the set of EU preferences that prefer $B_1(p)$ to $p$ is smaller than the set of EU preferences that prefer $B_2(p)$ to $p$. This is still true when we remove all the preferences that violate weak $FOSD$-monotonicity. We end up with the maximal sets of EU preferences, so the maximal set for DM 1 must be smaller than that for DM 2. The converse is obvious. 

\subsection{Proof of Proposition \ref{prop:unobserved_EU}}
Suppose that $(\succsim, c)$ has an EU justifiability representation for some preference $\succsim$. Suppose that the minimal set of justifiable preferences for $(\succsim, c)$ is not a singleton. Then, $\{q \in \Delta(Z): p \in c(\{p, q\})\}$ will not be a half-space. We show there is a unique EU preference such that $\{p, q\} = c(\{p, q\})$ implies $p \sim q$, and $B(p)$ is convex. First, consider the set of lotteries $q$ such that $\{p, q\} = c(\{p, q\})$. For at least one such $q$---call it $q^*$---there must be $m_1, m_2 \in \mathcal{M}^{EU}$ such that
\[m_1'q^* > m_1'p \text{ and } m_2'q^* < m_2'p.\]
Suppose not. Then all $m \in \mathcal{M}^{EU}$ must agree on the subset of $\{q \in \Delta(Z): p \sim q\}$ that is strictly worse than $p$. But then there is some $m^*$ that minimizes (in the set-inclusion sense) the subset of $\{q \in \Delta(Z): p \succsim q\}$ that is strictly worse than $p$. Any decision justified by some $m \in \mathcal{M}^{EU}$ will also be justified by $m^*$, so every justifiable preference but $m^*$ can be discarded. This contradicts the assumption that $\mathcal{M}^{EU}_\succ$ is non-singleton. 

Now consider perturbing $q^*$ to $\tilde{q}$ along $\{q \in \Delta(Z): p \sim q\}$. If $\tilde{q}$ is sufficiently close to $q^*$, $m_1$ and $m_2$ will continue to disagree, so we will have $\{\tilde{q}, p\} = c(\{\tilde{q}, p\})$ as well as $\tilde{q} \sim p$. By perturbing $q^*$ in each possible direction, we obtain a set of perturbations that pins down a unique hyperplane. This hyperplane is the indifference curve of the true preference. 

Now, we just need to resolve the direction of preference. Recall from the proof of Theorem \ref{thm:EU} that $B(p)$ is convex. Now consider
\[\{q \in \Delta(Z): p \precsim q \text{ and } p \notin c(\{p, q\})\}.\]
This set is obtained by reversing the direction of preference in the definition of $B(p)$. We need to show that it is \textit{not} convex. It suffices to show there are $q_1, q_2 \succ p$ and $\alpha \in (0, 1)$ such that $\{q_1\} = c(\{p, q_1\})$ and $\{q_2\} = c(\{p, q_2\})$ but $\{p\} = c(\{p, \alpha q_1 + (1-\alpha)q_2\})$. That is, $q_1, q_2 \notin W(p)$ but $\alpha q_1 + (1-\alpha)q_2 \in W(p)$. Take any two points $q_1, q_2 \succ p$ that lie on the boundary of $W(p)$, but don't lie on the same face of $W(p)$. (What if there are no such points? Then $W(p)$ must have only one face. This only happens if the minimal set of justifiable preferences is a singleton, which we have already ruled out.) Any linear combination of $q_1, q_2$ must belong to $W(p)$. This is exactly what we needed. The direction of preference is indeed pinned down by convexity of $B(p)$. Pinning down an indifference curve and the direction of preference is enough to pin down an EU preference, so we are done. 

\subsection{Proof of Theorem \ref{thm:unknown_true_pref}}
First, we need a bit of notation for cycles and chains. If $(a, b, d)$ is a cycle, write $a \; C \; b \; C \; d$. For any chain $(x_1, \ldots, x_k)$, we have $x_1 \; C \; x_2 \; \cdots \; x_{k-1} \; C \; x_k$. 

\begin{lemma}
\label{lem:construct_true_pref}
Define a binary relation $\succ$ as follows. First, say $a \succ b$ if $(a, b)$ is in the transitive closure of $C$. Second, if $a$ and $b$ have not yet been ranked, say $a \succ b$ if $a = c(\{a, b\})$, and $b \succ a$ otherwise. The result is a strict preference. 
\end{lemma}

\begin{proof}
Clearly, the result is complete. Suppose we have a cycle in $\succ$. We can write it $(x_1, \ldots, x_n)$ where $x_1 = x_n$. For each adjacent pair $(x_i, x_{i+1})$ there are two possibilities: (1) $(x_i, x_{i+1}) \in \text{tr}(C)$, or (2) $(x_i, x_{i+1}), (x_i,x_{i+1}) \notin \text{tr}(C)$ and $x_i = c(\{x_i, x_{i+1}\})$. 

There will be some $(x_i, x_{i+1}) \in \text{tr}(C)$. Otherwise, all the adjacent pairs would be ranked by in accordance with pairwise choice, so pairwise choice would have to be cyclic. But if pairwise choice is cyclic, there is a cycle involving three elements, at least two of which ($x_i$ and $x_{i+1}$) are adjacent. Suppose that $(x_{i+1}, x_i) \in \text{tr}(C)$, but $(x_i, x_{i+1}) \notin \text{tr}(C)$. This is inconsistent with $x_i \succ x_{i+1}$. We conclude that $(x_i, x_{i+1}) \in \text{tr}(C)$. We must have a sequence $(y_1, \ldots, y_k)$ such that $x_i \; C \; y_1 \; C \; \cdots \; y_k \; C \; x_{i+1}$ as well as $x_i = c(\{x_i, y_1\})$, $y_j = c(\{y_j, y_{j+1}\})$ and $y_k = c(\{y_k, x_{i+1}\})$. We can expand the cycle by adding $(y_1, \ldots, y_k)$ between $x_i$ and $x_{i+1}$. 

We have $x_i \; C \; y_1$. Suppose that $(x_{i-1}, x_i) \notin \text{tr}(C)$. If $y_1 = c(\{x_{i-1}, y_1\})$, we have a cycle involving $x_{i-1}, x_i, y_1$, so $x_{i-1}$ and $x_i$ must be ranked by $\text{tr}(C)$. But we have assumed that $(x_{i-1}, x_i) \notin \text{tr}(C)$. Given this, $(x_i, x_{i-1}) \in \text{tr}(C)$ is inconsistent with $x_{i-1} \succ x_i$. We conclude that $x_{i-1} = c(\{x_{i-1}, y_1\})$. Consider deleting $x_i$. We need to make sure $x_{i-1} \succ y_1$. Suppose we only have $y_1 \succ x_{i-1}$. This is consistent with $x_{i-1} = c(\{x_{i-1}, y_1\})$ only if $(y_1, x_{i-1}) \in \text{tr}(C)$. But since $(x_i, y_1) \in C$, this implies $(x_i, x_{i-1}) \in \text{tr}(C)$. We have already ruled this out. We have $x_{i-1} \succ y_1$ as desired: deleting $x_i$ didn't break the cycle. 

By proceeding in this way, we can get rid of all the adjacent pairs in the cycle that don't belong to $\text{tr}(C)$. The procedure will finish at some point because we started with a finite number of such pairs, and we aren't introducing any new ones. (We are only introducing pairs of items ranked by $C$.) We end up with a cycle in which (1) $x_i \; C \; x_{i+1}$ and (2) $x_i = c(\{x_i, x_{i+1}\})$ for every adjacent pair $(x_i, x_{i+1})$. By (1), every pair of items in the cycle (adjacent or not) is in $\text{tr}(C)$. In particular, $(x_{i+1}, x_i) \in \text{tr}(C)$ for every adjacent pair $(x_i, x_{i+1})$. Since $x_i = c(\{x_i, x_{i+1}\})$, $x_{i+1}$ is revealed excluded by $x_i$. Consider the menu consisting of all the items in the expanded cycle. By IEA, all the items in the menu can be removed (at once) without affecting choice. Since something must be chosen from the menu, we have a contradiction. 
\end{proof}

We let $\mathcal{M}$ be the set of strict preferences consistent with revealed exclusion. That is, $\succ_m$ belongs to $\mathcal{M}$ if and only if
\[a \text{ is revealed excluded by } B \quad \Longrightarrow \quad b \succ_m a \text{ for some } b \in B.\]

It remains to show that $(\succ, \mathcal{M})$ deliver the correct predictions. First, suppose that $c(A) \notin M(A)$: there is no $\succ_m$ such that $c(A) \succ_m A \setminus \{c(A)\}$. To see why this doesn't happen, recall the tree construction from the proof of Theorem \ref{thm:main}.

\begin{lemma}
\label{lem:exc2tree}
If no $\succ_m \in \mathcal{M}$ has $y \succ_m X$, then there is a revealed-exclusion-tree starting at $y$ and ending at $X$. 
\end{lemma}

\begin{proof}
We construct a candidate $\succ_m$ such that $y \succ X$. Let $B_0 = X$. Let $B_i$ be the union of $B_{i-1}$ and the set of items revealed excluded by subsets of $B_{i-1}$. Eventually, we reach $I$ such that $B_I = B_{I-1}$. Call this final set $B$. Rank all the items in $B$ by reverse-$\succ$ order. Let $T = \mathcal{A} \setminus B$. Rank all the items in $T$ by reverse-$\succ$ order. Finally, impose $t \succ_m b$ for all $t \in T, b \in B$. 

Notice that $y \in T$ unless there is a revealed-exclusion-tree starting at $y$ and ending at $X$. It remains to show that $\succ_m \in \mathcal{M}^L$. Suppose that $t \in T$ is revealed excluded by $T' \cup B'$, where $T' \subset T$ and $B' \subseteq B$. We must have $t \succ T'$, which implies $T' \succ_m t$. There is no problem unless $T'$ is empty. But then $t$ is revealed excluded by a subset of $B$, so must be in $B$---contradiction. Now suppose that $b \in B$ is revealed excluded by $B' \subset B$. We must have $b \succ B'$, which implies $B' \succ_m b$. There is no problem. $\succ_m$ is consistent with revealed exclusion. 
\end{proof}

By Lemma \ref{lem:exc2tree}, there must be a tree starting at $c(A)$ and ending at $A' \subseteq A \setminus \{c(A)\}$ in which each node that isn't in $A'$ is revealed excluded by its parents. Consider the set consisting of all the nodes in the tree, plus anything else in $A$. Since $c(A)$ is revealed excluded by its parents, IEA implies that $c(A)$ is not chosen from this set. Now remove everything that is not in $A$. IEA says that choice is unchanged. That is, $c(A)$ can't be chosen from $A$---contradiction. 

Only one more thing could go wrong: $a, c(A) \in M(A)$ and $a \succ c(A)$. 

\begin{lemma}
\label{lem:exc2pref}
If $a \in A$ is not revealed excluded by any subset of $A$, and if $a \neq c(A)$, then $c(A) \succ a$.
\end{lemma}

\begin{proof}
We start by removing the items in $A$ that are revealed excluded by a subset of $A$. Call the resulting set $A^*$. By assumption, $a$ is not removed. By IEA, choice is unchanged: $c(A) = c(A^*)$ Take any $A' \subseteq A^*$ such that $|A'| = k$. Suppose that choice on the proper subsets of $A'$ satisfies WARP. (This is trivially true if $|A'| = 3$. If choice on $A'$ violates WARP, then $A'$ is an almost-WARP set or a cycle. Either way, some item in $A'$ is revealed excluded by a subset of $A'$, which contradicts the definition of $A^* \supseteq A$. We conclude that choice on $A^*$ satisfies WARP, so $c(A) = c(A^*)$ is chosen whenever it is available. In particular, $c(A) = c(\{c(A), a\})$. This implies $c(A) \succ a$ unless $(a, c(A)) \in \text{tr}(C)$. But then $a$ is revealed excluded by $c(A) \in A$, which contradicts the assumption about $a$. Thus, $c(A) \succ a$.
\end{proof}

If $a \in M(A)$, then $a$ is not revealed excluded by any subset of $A$. The Lemma delivers a contradiction. 

\subsection{Proof of Corollary \ref{cor:maximal}}
The proof of Theorem \ref{thm:unknown_true_pref} constructs precisely this representation. It is obvious that $\mathcal{M}$ is maximal. If any representation has a justifiable preference that is not in $\mathcal{M}$, it has a justifiable preference that violates revealed exclusion. That is, $a \succ_m B$ even though $a$ is revealed excluded by $B$. We know that $a \succ B$ in every representation, so this representation predicts $c(B \cup \{a\}) = a$. By IEA, this cannot be the case. 

For uniqueness of $\succ$, consider some representation $(\succ', \mathcal{M})$ such that $a \succ b$ but $b \succ' a$. We know that all representations agree on pairs of items in $\text{tr}(C)$. Thus, $(a, b), (b, a) \notin \text{tr}(C)$, and $c(\{a, b\}) = a$. To get the right prediction, the second representation cannot have any justifiable preference $\succ_m$ such that $b \succ_m a$. Notice that $\mathcal{M}$ contains a preference $\succ_{bad}$ that is exactly opposite $\succ$. $\succ_{bad}$ is consistent with revealed exclusion because $a \succ B$, so $B \succ_{bad} a$, whenever $a$ is revealed excluded by $B$. Thus, $\mathcal{M}$ contains at least one $\succ_m$ such that $a \succ_m b$. $(\succ', \mathcal{M})$ is not a representation---contradiction. 

\subsection{Proof of Proposition \ref{prop:rev_exc}}
Suppose there is no justifiable preference in any representation that ranks $a$ over $B$. Thus, there is no justifiable preference in the canonical representation that ranks $a$ over $B$. By Lemma \ref{lem:exc2tree}, there is a revealed-exclusion-tree starting at $a$ and ending at $A \subseteq B$. Since $y$ is revealed excluded by $X$ only if $y \succ X$, each node must be ranked strictly above its parents. Thus, we have $a \succ A$ as well as $a \neq c(\{a\} \cup A)$. 

We can require $A$ to be minimal, meaning $a = c(\{a\} \cup A')$ for any proper subset $A'$ of $A$. We show that $a$ is revealed excluded by $A$. 

First, choice on all the proper subsets of $A \cup \{a\}$ must satisfy WARP. Suppose it doesn't. Then, we can find some cycle or almost-WARP set $A' \subset A \cup \{a\}$. We conclude that some item in $A \cup \{a\}$ is revealed excluded by some subset of $A \cup \{a\}$. Since $a = c(\{a\} \cup A')$ for every $A' \subset A$, $a$ cannot be revealed excluded. Suppose $a' \neq a$ is revealed excluded. By IEA, removing $a'$ from $\{a\} \cup A$ has no effect on choice. Since $a \neq c(\{a\} \cup A)$, $a \neq c(\{a\} \cup A \setminus \{a'\})$. This contradicts minimality of $A$. We conclude that $A \cup \{a\}$ is almost-WARP. Since $a$ is chosen from every proper subset containing it, $a$ is revealed excluded by $A$. 

Now suppose that $A = \{a'\}$. Recall that $a \succ a'$ and $c(\{a, a'\}) = a'$. This happens in the canonical representation only if $(a, a') \in \text{tr}(C)$. In that case, $a$ is revealed excluded by $a'$. 

\subsection{Proof of Proposition \ref{prop:double}}

Since $c_L$ satisfies IEA, we construct $\succ$ in accordance with Lemma \ref{lem:construct_true_pref}. We then let $\mathcal{M}^L$ be the set of strict preferences consistent with revealed exclusion in $L$. That is, $\succ_m \in \mathcal{M}^L$ if and only if
\[a \text{ is revealed excluded by } B \text{ in } L \quad \Longrightarrow \quad b \succ_m a \text{ for some } b \in B.\]
For $\mathcal{M}^H$, we need to define a new relation $R$ that captures replacement as well as revealed exclusion.

\begin{definition}[Relation $R$]
Say that $Z \; R \; z$ if either of the following holds:
\begin{enumerate}
\item $z$ is revealed excluded in $L$ by $Z$.
\item $z$ is replaced in $Z \cup \{z\}$, and no item in $Z$ is revealed excluded in $L$ by any subset of $Z \cup \{z\}$.
\end{enumerate}
\end{definition}

Let $\mathcal{M}^H$ be the set of strict preferences consistent with $R$. That is, $\succ_m \in \mathcal{M}^H$ if and only if
\[B \; R \; a \quad \Longrightarrow \quad b \succ_m a \text{ for some } b \in B.\]
Notice that each $\succ_m \in \mathcal{M}^H$ is consistent with replacement:
\[a \text{ is replaced in B} \quad \Longrightarrow \quad b \succ_m a \text{ for some } b \in B \setminus \{a\}.\]
To see why, suppose $a = c_L(B) \neq c_H(B)$. Remove any items in $B$ that are revealed excluded in $L$ by some subset of $B$. Call the result $B^*$. Since $c_L$ satisfies IEA and $c_H$ satisfies IREA, we have $a = c_L(B) = c_L(B^*)$ and $c_H(B) = c_H(B^*)$. Thus, $a$ is replaced in $B^*$ as well as $B$. For $\succ_m \in \mathcal{M}^H$, there will be some $b \in B^* \subseteq B$ such that $b \succ_m a$. Since no item in $B^*$ is revealed excluded in $L$ by any subset of $B^*$, this is exactly what we needed. 

It remains to show that $(\succ, \mathcal{M}^L)$ and $(\succ, \mathcal{M}^H)$ deliver the correct predictions. Suppose $c_L(A) \notin M_L(A)$. The argument from the proof of Theorem \ref{thm:unknown_true_pref} rules this out---just replace ``revealed excluded'' with ``revealed excluded in $L$.'' Now suppose $c_H(A) \notin M_H(A)$. A similar argument rules this out, but we need to modify the tree structure. Instead of requiring each node to be revealed excluded in $L$ by its parents, we require each node's parents to stand in relation $R$ to it. We need a result analogous to Lemma \ref{lem:exc2tree}. 

\begin{lemma}
If no $\succ_m \in \mathcal{M}^H$ has $y \succ_m X$, then there is an $R$-tree starting at $y$ and ending at $X$. 
\end{lemma}

\begin{proof}
The proof is similar to that of Lemma \ref{lem:exc2tree}. Now we let $B_i$ be the union of $B_{i-1}$ and the set of items $b$ such that $B' \; R \; b$ for some $B' \subseteq B_{i-1}$. The rest of the construction is as before. So is the proof that $\succ_m \in \mathcal{M}^H$, with one exception. We have to establish that $z \succ Z$ whenever $Z \; R \; z$. We already know this is true if $z$ is revealed excluded by $Z$. Suppose that $z$ is replaced in $Z \cup \{z\}$, so $z = c_L(Z \cup \{z\}) \neq c_H(Z \cup \{z\})$. Suppose further that no item in $Z$ is revealed excluded in $L$ by any subset of $Z \cup \{z\}$. 

Fix $z' \in Z$. There are two ways to get $z' \succ z$. The first is $c_L(\{z', z\}) = z$ and $(z', z) \in \text{tr}(C_L)$. In this case, $z'$ is revealed excluded in $L$ by $z$, which contradicts our assumption about $Z$. The second is $c_L(\{z', z\}) = z'$. We show this doesn't happen: $c_L(\{z, z'\}) = c_L(Z \cup \{z\}) = z$. Take $Z' \subseteq Z \cup \{z\}$ such that $|Z'| = k$. Suppose that the restriction of $c_L$ to the proper subsets of $Z'$ satisfies WARP. (This is trivially true if $|Z'| = 3$.) If the restriction of $c_L$ to $Z'$ violates WARP, then $Z'$ is an almost-WARP set or a cycle. In either case, some item in $Z'$ is revealed excluded in $L$ by the rest of $Z'$, which contradicts our assumption about $Z$. We conclude that the restriction of $c_L$ to $Z \cup \{z\}$ satisfies WARP, so $z = c_L(\{z\} \cup Z)$ is chosen whenever it is available. In particular, $c_L(\{z, z'\}) = z$ for all $z' \in Z$. This completes the proof that $Z \; R \; z$ implies $z \succ Z$. The argument from the proof of Lemma \ref{lem:exc2tree} goes through from here. 
\end{proof}

We must have an $R$-tree starting at $c_H(A)$ and ending at $A' \subset A \setminus \{c_H(A)\}$. Consider the set consisting of all the nodes in the tree, plus anything else in $A$. IREA implies that $c_H(A)$ is not chosen from this set under $H$. Now remove everything that is not in $A$. IREA says that choice under $H$ is unchanged. That is, $c_H(A)$ can't be chosen from $A$ under $H$---contradiction. 

Now suppose $a, c_L(A)$ in $M_L(A)$ and $a \succ c_L(A)$. Lemma \ref{lem:exc2pref} rules this out just as before. Finally, suppose $a, c_H(A)$ in $M_H(A)$ and $a \succ c_H(A)$. We can remove all the items in $A$ that are revealed excluded in $L$ by, or replaced in, a subset of $A$. Call the resulting set $A^*$. By IREA, choice in $H$ is unchanged: $c_H(A) = c_H(A^*)$. Notice that $a \in A^*$: otherwise, $a$ could not be in $M_H(A)$. Notice also that $c_L$ and $c_H$ agree on $A^*$. (Otherwise, some items in $A^*$ would be replaced in a subset of $A$, contradicting the definition of $A^*$.) We can now apply Lemma \ref{lem:exc2pref} to get $c_L(A^*) \succ a$. Since $c_L(A^*) = c_H(A^*)$ and $c_H(A^*) = c_H(A)$, we have $c_H(A) \succ a$---contradiction.

\subsection{Proof of Corollary \ref{cor:WARP_for_cL}}
Suppose that $c_L$ satisfies WARP, so it maximizes a unique preference $\succ$. If $c_H$ has a justifiability representation, then $(c_L, c_H)$ satisfies the conditions in Proposition \ref{prop:double}. We verify that $c_H$ will also satisfy IUA conditional on $\succ$. It suffices to show that $c_H$ has a justifiability representation with true preference $\succ$. But since there are no cycles under $L$, this is precisely the representation constructed in Proposition \ref{prop:double}.

\section{Additional Results}

\subsection{Case excluded from Proposition \ref{prop:unobserved_EU}}

Unlike Proposition \ref{prop:unobserved_EU}, the first part of Proposition \ref{prop:edge_case} is completely general. It applies whenever $c$ has an EU justifiability representation, even if the minimal set of justifiable EU preferences is a singleton. 

\begin{proposition}
\label{prop:edge_case}
Fix $c$ and $p \in \text{int}(\Delta(Z))$. 
\begin{enumerate}
\item Either $c$ lacks a justifiability representation, or it has one for precisely those EU preferences such that
\begin{align}
\label{eq:indiff}
&\{p, q\} = c(\{p, q\}) \text{ implies } p \sim q \\
\label{eq:direction}
\text{and } &\{q \in \Delta(Z): p \in c(\{p, q\}) \text{ and } p \succsim q\} \text{ is closed}.
\end{align}
\item Unless $\{q \in \Delta(Z): p \in c(\{p, q\})\}$ is a half-space (restricted to the simplex), there is at most one EU preference that satisfies (\ref{eq:indiff}) and (\ref{eq:direction}).
\end{enumerate}
\end{proposition}

\subsection{Comparing minimal sets of justifiable EU preferences}

\begin{corollary}
Suppose that $(\succsim, c_1)$ and $(\succsim, c_2)$ have EU justifiability representations $(u, \mathcal{M}_1^{EU})$ and $(u, \mathcal{M}_2^{EU})$ respectively. If $B_1(p) \supseteq B_2(p)$ for some $p \in \text{int}(\delta(Z))$, then $(u, \text{co}(\mathcal{M}^{EU}_1 \cup \mathcal{M}^{EU}_2))$ is an EU justifiability representation of $(\succsim, c_2).$
\end{corollary}

\begin{proof}
Construct the justifiable Bernoulli utilities for each DM as in Theorem \ref{thm:EU}. Then take the convex hull of the union. We will have
\begin{equation}
\label{eq:liberal}
B_2(p) = \bigcap_{m \in \text{co}(\mathcal{M}^{EU}_1 \cup \mathcal{M}^{EU}_2)}\{q \in \Delta(Z): m'q > m'p\} \cap \{q \in \Delta(Z): p \succsim q\}
\end{equation}
In the proof of Theorem \ref{thm:EU}, we saw that (\ref{eq:liberal}) holds with $\mathcal{M}^{EU}_2$ in place of $\text{co}(\mathcal{M}^{EU}_1 \cup \mathcal{M}^{EU}_2)$. Thus, (\ref{eq:liberal}) fails only if there exist $m \in \mathcal{M}^{EU}_1$ and $q \in B_2(p)$ such that $m'p \geq m'q$. Since 
\[B_1(p) = \bigcap_{m \in \mathcal{M}^{EU}_1}\{q \in \Delta(Z): m'q > m'p\} \cap \{q \in \Delta(Z): p \succsim q\},\]
$q$ cannot belong to $B_1(p)$. This contradicts $B_1(p) \supseteq B_2(p)$, so (\ref{eq:liberal}) holds. This is all we need to establish that $\text{co}(\mathcal{M}^{EU}_1 \cup \mathcal{M}^{EU}_2)$ is a set of justifiable Bernoulli utilities for DM 2. 
\end{proof}

\end{document}